\documentclass[11pt]{article}%
\usepackage{amsmath}
\usepackage{amssymb}
\usepackage{graphicx}
\usepackage{fullpage}
\usepackage{amsfonts}%
\providecommand{\U}[1]{\protect\rule{.1in}{.1in}}
\newtheorem{theorem}{Theorem}

\newtheorem{conjecture}[theorem]{Conjecture}
\newtheorem{corollary}[theorem]{Corollary}

\newtheorem{definition}[theorem]{Definition}

\newtheorem{lemma}[theorem]{Lemma}

\newtheorem{problem}[theorem]{Problem}
\newtheorem{proposition}[theorem]{Proposition}

\newenvironment{proof}[1][Proof]{\noindent\textbf{#1.} }{\ \rule{0.5em}{0.5em}}
\begin{document}

\title{The Need for Structure in Quantum Speedups}
\author{Scott Aaronson\thanks{MIT. \ Email: aaronson@csail.mit.edu. \ This material is
based upon work supported by the National Science Foundation under Grant No.
0844626,\ by a TIBCO Career Development Chair, and by an Alan T. Waterman
award.}\\MIT
\and Andris Ambainis\thanks{Email: ambainis@lu.lv. \ Supported by University of
Latvia Research Grant ZP01-100, FP7 Marie Curie International Reintegration
Grant (PIRG02-GA-2007-224886), FP7 FET-Open project QCS and ERC Advanced Grant
MQC (at the University of Latvia) and the National Science Foundation under
agreement No. DMS-1128155 (at IAS, Princeton). Any opinions, findings and
conclusions or recommendations expressed in this material are those of the
author(s) and do not necessarily reflect the views of the National Science
Foundation.}\\University of Latvia and IAS, Princeton}
\date{}
\maketitle

\begin{abstract}
Is there a general theorem that tells us when we can hope for exponential
speedups from quantum algorithms, and when we cannot? \ In this paper, we make
two advances toward such a theorem, in the black-box model where most quantum
algorithms operate.

First, we show that for any problem that is invariant under permuting inputs
and outputs and that has sufficiently many outputs (like the collision and
element distinctness problems), the quantum query complexity is at least the
$7^{th}$ root of the classical randomized query complexity. \ (An earlier
version of this paper gave the $9^{th}$\ root.) \ This resolves a conjecture
of Watrous from 2002.

Second, inspired by work of O'Donnell et al.\ (2005) and Dinur et al.\ (2006),
we conjecture that every bounded low-degree polynomial has a \textquotedblleft
highly influential\textquotedblright\ variable. \ Assuming this conjecture, we
show that every $T$-query quantum algorithm can be simulated on \textit{most}
inputs by a $T^{O(1)}$-query\ classical algorithm,\ and that one essentially
cannot hope to prove $\mathsf{P}\neq\mathsf{BQP}$ relative to a random oracle.

\end{abstract}

\section{Introduction\label{INTRO}}

Perhaps the central lesson gleaned from fifteen years of quantum algorithms
research is this:

\begin{quotation}
\noindent\textit{Quantum computers can offer superpolynomial speedups over
classical computers, but only for certain \textquotedblleft
structured\textquotedblright\ problems}.
\end{quotation}

The key question, of course, is what we mean by \textquotedblleft
structured.\textquotedblright\ \ In the context of most existing quantum
algorithms, \textquotedblleft structured\textquotedblright\ basically means
that we are trying to determine some global property of an extremely long
sequence of numbers, \textit{assuming} that the sequence satisfies some global
regularity. \ As a canonical example, consider \textsc{Period-Finding}, the
core of Shor's algorithms for factoring and computing discrete logarithms
\cite{shor}. \ Here we are given black-box access to an exponentially-long
sequence of integers $X=(x_{1},\ldots,x_{N})$; that is, we can compute $x_{i}%
$\ for a given $i$. \ We\ are asked to find the \textit{period} of $X$---that
is, the smallest $k>0$ such that $x_{i}=x_{i-k}$\ for all $i>k$---promised
that $X$ is indeed periodic, with period $k\ll N$ (and also that the $x_{i}%
$\ values are approximately distinct within each period). \ The requirement of
periodicity is crucial here: it is what lets us use the Quantum Fourier
Transform to extract the information we want from a superposition of the form%
\[
\frac{1}{\sqrt{N}}\sum_{i=1}^{N}\left\vert i\right\rangle \left\vert
x_{i}\right\rangle .
\]
For other known quantum algorithms, $X$ needs to be (for example) a cyclic
shift of quadratic residues \cite{dhi}, or constant on the cosets of a hidden subgroup.

By contrast, the canonical example of an \textquotedblleft
unstructured\textquotedblright\ problem is the Grover search problem. \ Here
we are given black-box access to an $N$-bit string $(x_{1},\ldots,x_{N}%
)\in\left\{  0,1\right\}  ^{N}$, and are asked whether there exists an $i$
such that $x_{i}=1$.\footnote{A variant asks us to \textit{find} an $i$ such
that\ $x_{i}=1$, under the mild promise that such an $i$ exists.} \ Grover
\cite{grover}\ gave a quantum algorithm to solve this problem using
$O(\sqrt{N})$\ queries \cite{grover}, as compared to the $\Omega\left(
N\right)  $\ needed classically. \ However, this quadratic speedup is optimal,
as shown by Bennett, Bernstein, Brassard, and Vazirani \cite{bbbv}. \ For
other \textquotedblleft unstructured\textquotedblright\ problems---such as
computing the \textsc{Parity} or \textsc{Majority} of an $N$-bit
string---quantum computers offer no asymptotic speedup at all over classical
computers (see Beals et al.\ \cite{bbcmw}).

Unfortunately, this \textquotedblleft need for structure\textquotedblright%
\ has essentially limited the prospects for superpolynomial quantum speedups
to those areas of mathematics that are liable to produce things like periodic
sequences or sequences of quadratic residues.\footnote{Here we exclude
$\mathsf{BQP}$-complete problems, such as simulating quantum physics (the
\textquotedblleft original\textquotedblright\ application of quantum
computers), approximating the Jones polynomial \cite{ajl}, and estimating a
linear functional of the solution of a well-conditioned linear system
\cite{hhl}.} \ \textit{This is the fundamental reason why the greatest
successes of quantum algorithms research have been in cryptography, and
specifically in number-theoretic cryptography.} \ It helps to explain why we
do not have a fast quantum algorithm to solve $\mathsf{NP}$-complete problems
(for example), or to break arbitrary one-way functions.

Given this history, the following problem takes on considerable importance:

\begin{problem}
[Informal]\label{theprob}For every \textquotedblleft
unstructured\textquotedblright\ problem $f$, are the quantum query complexity
$\operatorname*{Q}(f)$\ and the classical randomized query complexity
$\operatorname*{R}(f)$ polynomially related\textbf{?}
\end{problem}

Despite its apparent vagueness, Problem \ref{theprob}\ can be formalized in
several natural and convincing ways---and under these formalizations, the
problem has remained open for about a decade.

\subsection{Formalizing the Problem\label{FORMAL}}

Let $S\subseteq\left[  M\right]  ^{N}$\ be a collection of inputs, and let
$f:S\rightarrow\left\{  0,1\right\}  $\ be a function that we are trying to
compute. \ In this paper, we assume for simplicity that the range of $f$\ is
$\left\{  0,1\right\}  $; in other words, that we are trying to solve a
decision problem. \ It will also be convenient to think of $f$\ as a function
from $\left[  M\right]  ^{N}$\ to $\left\{  0,1,\ast\right\}  $, where $\ast
$\ means `undefined' (that is, that a given input $X\in\left[  M\right]  ^{N}%
$\ is not in $f$'s domain $S$).

We will work in the well-studied \textit{decision-tree model}. \ In this
model, given an input $X=(x_{1},\ldots,x_{N})$, an algorithm can at any time
choose an $i$ and receive $x_{i}$.\ \ We count only the number of queries the
algorithm makes to the $x_{i}$'s, ignoring other computational steps. \ Then
the \textit{deterministic query complexity} of $f$, or $\operatorname*{D}(f)$,
is the number of queries made by an optimal deterministic algorithm on a
worst-case input $X\in S$. \ The (bounded-error)\textit{ randomized query
complexity} $\operatorname*{R}(f)$\ is the expected number of queries made by
an optimal randomized algorithm that, for every $X\in S$, computes
$f(X)$\ with probability at least $2/3$. \ The (bounded-error)\textit{ quantum
query complexity} $\operatorname*{Q}(f)$\ is the same as $\operatorname*{R}%
(f)$, except that we allow quantum algorithms. \ Clearly $\operatorname*{Q}%
(f)\leq\operatorname*{R}(f)\leq\operatorname*{D}(f)\leq N$\ for all $f$. \ See
Buhrman and de Wolf \cite{bw}\ for detailed definitions as well as a survey of
these measures.

If $S=\left[  M\right]  ^{N}$, then we say $f$ is \textit{total}, and if
$M=2$, then we say $f$ is \textit{Boolean}.\ \ The case of total $f$ is
relatively well-understood. \ Already in 1998, Beals et al.\ \cite{bbcmw}%
\ showed the following:

\begin{theorem}
[Beals et al.]\label{bbcmwthm}$\operatorname*{D}\left(  f\right)
=O(\operatorname*{Q}(f)^{6})$ for all total Boolean functions $f:\left\{
0,1\right\}  ^{N}\rightarrow\left\{  0,1\right\}  $.
\end{theorem}

Furthermore, it is easy to generalize Theorem \ref{bbcmwthm}\ to show
that\ $\operatorname*{D}\left(  f\right)  =O(\operatorname*{Q}(f)^{6})$\ for
\textit{all} total functions $f:\left[  M\right]  ^{N}\rightarrow\left\{
0,1\right\}  $, not necessarily Boolean.\footnote{Theorem \ref{bbcmwthm}\ is
proved by combining three ingredients: $\operatorname*{D}(f)=O\left(
\operatorname*{C}(f)\operatorname*{bs}(f)\right)  $, $\operatorname*{C}%
(f)=O(\operatorname*{bs}(f)^{2})$, and $\operatorname*{bs}%
(f)=O(\operatorname*{Q}(f)^{2})$\ (where $\operatorname*{C}(f)$\ is the
\textit{certificate complexity} of $f$ and $\operatorname*{bs}(f)$\ is the
\textit{block sensitivity}). \ And all three ingredients go through with no
essential change if we set $M>2$, and define suitable $M$-ary generalizations
of $\operatorname*{C}(f)$\ and $\operatorname*{bs}(f)$. \ (We could also
convert the non-Boolean function $f:\left[  M\right]  ^{N}\rightarrow\left\{
0,1\right\}  $\ to a Boolean one, but then we would lose a factor of $\log
M$.)} \ In other words, for total functions, the quantum query complexity is
always at least the $6^{th}$\ root of the classical query complexity. \ The
largest known gap between $\operatorname*{D}(f)$\ and $\operatorname*{Q}%
(f)$\ for a total function is quadratic, and is achieved by the
$\operatorname*{OR}$\ function (because of Grover's algorithm).

On the other hand, as soon as we allow non-total functions, we can get
enormous gaps.\ \ Aaronson \cite{aar:ph} gave a Boolean function
$f:S\rightarrow\left\{  0,1\right\}  $ for which $\operatorname*{R}%
(f)=N^{\Omega(1)}$, yet $\operatorname*{Q}(f)=O\left(  1\right)
$.\footnote{Previously, de Beaudrap, Cleve, and Watrous \cite{beaudrap}\ had
stated a similar randomized versus quantum separation. \ However, their
separation applied not to the standard quantum black-box model, but to a
different model in which the black box permutes the answer register
$\left\vert y\right\rangle $ in some unknown way (rather than simply mapping
$\left\vert y\right\rangle $\ to $\left\vert y\oplus f\left(  x\right)
\right\rangle $).} \ Other examples, for which $\operatorname*{R}%
(f)=\Omega(\sqrt{N})$\ and $\operatorname*{Q}(f)=O(\log N\log\log N)$,\ follow
easily from Simon's algorithm \cite{simon}\ and Shor's algorithm \cite{shor}.
\ Intuitively, these functions $f$\ achieve such large separations by being
highly structured: that is, their domain $S$ includes only inputs that satisfy
a stringent promise, such as encoding a periodic function, or (in the case of
\cite{aar:ph}) encoding two Boolean functions, one of which is correlated with
the Fourier transform of the other one.

By contrast with these highly-structured problems, consider the
\textit{collision problem}: that of deciding whether a sequence of numbers
$(x_{1},\ldots,x_{N})\in\left[  M\right]  ^{N}$\ is one-to-one (each number
appears once) or two-to-one (each number appears twice). \ Let \textsc{Col}%
$\left(  X\right)  =0$\ if $X$ is one-to-one and \textsc{Col}$\left(
X\right)  =1$\ if $X$ is two-to-one, promised that one of these is the case.
\ Then \textsc{Col}$\left(  X\right)  $ is not a total function, since its
definition involves a promise on $X$. \ Intuitively, however, the collision
problem seems much less \textquotedblleft structured\textquotedblright\ than
Simon's and Shor's problems. \ One way to formalize this intuition is as
follows. \ Call a partial function $f:\left[  M\right]  ^{N}\rightarrow
\left\{  0,1,\ast\right\}  $ \textit{permutation-invariant }if%
\[
f(x_{1},\ldots,x_{N})=f(\tau(x_{\sigma\left(  1\right)  }),\ldots
,\tau(x_{\sigma\left(  N\right)  }))
\]
for all inputs $X\in\left[  M\right]  ^{N}$\ and all permutations $\sigma\in
S_{N}$ and $\tau\in S_{M}$. \ Then \textsc{Col}$\left(  X\right)  $ is
permutation-invariant: we can permute a one-to-one sequence and relabel its
elements however we like, but it is still a one-to-one sequence, and likewise
for a two-to-one sequence. \ Because of this symmetry, attempts to solve the
collision problem using (for example) the Quantum Fourier Transform seem
unlikely to succeed. \ And indeed, in 2002 Aaronson \cite{aar:col}\ proved
that $\operatorname*{Q}\left(  \text{\textsc{Col}}\right)  =\Omega(N^{1/5})$:
that is, the quantum query complexity of the collision problem is at most
polynomially better than its randomized query complexity of $\Theta(\sqrt{N}%
)$. \ The quantum lower bound was later improved to $\Omega(N^{1/3})$\ by
Aaronson and Shi \cite{as}, matching an upper bound of Brassard, H\o yer, and
Tapp \cite{bht}.

Generalizing boldly from this example, John Watrous (personal communication)
conjectured that the randomized and quantum query complexities are
polynomially related for \textit{every} permutation-invariant problem:

\begin{conjecture}
[Watrous 2002]\label{watrousconj}$\operatorname*{R}(f)\leq\operatorname*{Q}%
(f)^{O\left(  1\right)  }$\ for every partial function $f:\left[  M\right]
^{N}\rightarrow\left\{  0,1,\ast\right\}  $ that is permutation-invariant.
\end{conjecture}

Let us make two remarks about Conjecture \ref{watrousconj}. \ First, the
conjecture\ talks about \textit{randomized} versus quantum query complexity,
since in this setting, it is easy to find functions $f$ for
which\ $\operatorname*{R}(f)$\ and $\operatorname*{Q}(f)$\ are both tiny but
$\operatorname*{D}(f)$\ is huge. \ As an example, consider the
\textit{Deutsch-Jozsa problem} \cite{dj}: given a Boolean input $(x_{1}%
,\ldots,x_{N})$, decide whether the $x_{i}$'s are all equal or whether half of
them are $1$ and the other half are $0$, under the promise that one of these
is the case.

Second, if $M=2$\ (that is, $f$ is Boolean), then Conjecture \ref{watrousconj}%
\ follows relatively easily from known results:\ indeed, we prove in Appendix
\ref{BOOLEAN}\ that $\operatorname*{R}(f)=O(\operatorname*{Q}(f)^{2})$ in that
case. \ So the interesting case is when $M\gg2$, as it is for the collision problem.

Conjecture \ref{watrousconj} provides one natural way to formalize the idea
that classical and quantum query complexities should be polynomially related
for all \textquotedblleft unstructured\textquotedblright\ problems. \ A
different way is provided by the following conjecture, which we were aware of
since about 1999:

\begin{conjecture}
[folklore]\label{folkloreconj}Let $Q$ be a quantum algorithm that makes $T$
queries to a Boolean input $X=(x_{1},\ldots,x_{N})$, and let $\varepsilon>0$.
\ Then there exists a deterministic classical algorithm that makes
$\operatorname*{poly}(T,1/\varepsilon,1/\delta)$\ queries to the $x_{i}%
$'s,\ and that approximates $Q$'s acceptance probability to within an additive
error $\varepsilon$ on a $1-\delta$\ fraction of inputs.
\end{conjecture}

But what exactly does Conjecture \ref{folkloreconj} have to do with
\textquotedblleft the need for structure in quantum speedups\textquotedblright%
? \ With Conjecture \ref{watrousconj}, the connection to this paper's theme
was more-or-less obvious, but with Conjecture \ref{folkloreconj}, some
additional explanation is probably needed.

Intuitively, we want to say the following: in order to achieve a
superpolynomial speedup in the black-box model, a quantum computer needs not
merely a promise problem, but a \textquotedblleft severely
constrained\textquotedblright\ promise problem. \ In other words, only a
\textit{minuscule fraction} of the $2^{N}$\ oracle strings $X=(x_{1}%
,\ldots,x_{N})$ ought to satisfy the promise---precisely like what happens in
Simon's and Shor's problems, where the promise asserts that $X$ encodes a
periodic function. \ If the promise is too \textquotedblleft
mild\textquotedblright---if, say, it holds for all $X$ in some set
$S\subseteq\left\{  0,1\right\}  ^{N}$\ with $\left\vert S\right\vert
=\Omega(2^{N})$---then we should be back in the situation studied by Beals et
al.\ \cite{bbcmw}, where the oracle $X$\ lacked enough \textquotedblleft
structure\textquotedblright\ for a Shor-like algorithm to exploit, and as a
result, the best one could hope for was a \textit{polynomial} quantum speedup
like that of Grover's algorithm.

Yet, if we interpret the above intuition too na\"{\i}vely, then it is easy to
find counterexamples. \ To illustrate, let $S_{1}$\ consist of all strings
$X\in\left\{  0,1\right\}  ^{N}$\ that encode valid inputs to Simon's problem,
let $S_{0}$\ consist of all $Y\in\left\{  0,1\right\}  ^{N}$\ that have
Hamming distance at least $N/10$\ from every $X\in S_{1}$, and let
$S=S_{0}\cup S_{1}$.\ \ Then define a Boolean function
$f_{\operatorname*{Simon}}:S\rightarrow\left\{  0,1\right\}  $\ by
$f_{\operatorname*{Simon}}(X)=1$\ for all $X\in S_{1}$, and
$f_{\operatorname*{Simon}}(X)=0$\ for all $X\in S_{0}$. \ As observed by
Buhrman et al.\ \cite{bfnr}\ (see also Ambainis and de Wolf \cite{aw:average}
and Hemaspaandra, Hemaspaandra, and Zimand \cite{hhz}), this \textquotedblleft
property-testing version of Simon's problem\textquotedblright\ achieves an
exponential separation between randomized and quantum query complexities:
$\operatorname*{R}(f_{\operatorname*{Simon}})=\Omega(\sqrt{N/\log N})$\ while
$\operatorname*{Q}(f_{\operatorname*{Simon}})=O(\log N)$. \ But the promise is
certainly \textquotedblleft mild\textquotedblright: indeed $\left\vert
S\right\vert \geq2^{N}-2^{cN}$ for some constant $c<1$.

On the other hand, examining this counterexample more closely\ suggests a way
to salvage our original intuition. \ For notice that there exists a fast,
deterministic classical algorithm that correctly evaluates
$f_{\operatorname*{Simon}}(X)$ on \textit{almost all inputs} $X\in S$: namely,
the algorithm that always outputs $0$! \ This algorithm errs only on the
minuscule fraction of inputs $X\in S$ that happen to belong to $S_{1}$.
\ Thus, we might conjecture that this points to a general phenomenon: namely,
whenever there exists a fast quantum algorithm to compute a Boolean function
$f:S\rightarrow\left\{  0,1\right\}  $ with $\left\vert S\right\vert
=\Omega\left(  2^{N}\right)  $, there also exists a fast classical algorithm
to compute $f(X)$\ on \textit{most} inputs $X\in S$. \ In Appendix
\ref{EQUIV}, we will prove that Conjecture \ref{folkloreconj} is equivalent to
this conjecture.

Indeed, Conjecture \ref{folkloreconj}\ readily implies a far-reaching
generalization of the result of Beals et al.\ \cite{bbcmw} stating that
$\operatorname*{D}(f)=O(\operatorname*{Q}(f)^{6})$\ for all total Boolean
functions $f$. \ In particular, define the $\varepsilon$\textit{-approximate
query complexity} of a Boolean function $f:\left\{  0,1\right\}
^{N}\rightarrow\left\{  0,1\right\}  $, or $\operatorname*{D}_{\varepsilon
}(f)$, to be the minimum number of queries made by a deterministic algorithm
that evaluates $f$\ correctly on at least a $1-\varepsilon$\ fraction of
inputs $X$. \ Likewise, let $\operatorname*{Q}_{\varepsilon}(f)$\ be the
minimum number of queries made by a quantum algorithm that evaluates
$f$\ correctly on at least a $1-\varepsilon$\ fraction of inputs. \ Then
Conjecture \ref{folkloreconj}\ implies that $\operatorname*{D}_{\varepsilon
}(f)$\ and $\operatorname*{Q}_{\delta}(f)$ are polynomially related for all
Boolean functions $f$ and all constants $\varepsilon>\delta>0$ independent of
$N$.\footnote{More generally, as we will show in Corollary \ref{dqcor}, the
relation we obtain is\ $\operatorname*{D}_{\varepsilon+\delta}(f)\leq\left(
\operatorname*{Q}_{\varepsilon}(f)/\delta\right)  ^{O\left(  1\right)  }$ for
all $\varepsilon,\delta>0$.} \ This would provide a quantum counterpart to a
beautiful 2002 result of Smyth \cite{smyth},\ who solved an old open problem
of Steven Rudich by showing that $\operatorname*{D}_{\varepsilon
}(f)=O(\operatorname*{C}_{\varepsilon^{3}/30}(f)^{2}/\varepsilon^{3})$\ for
all $\varepsilon>0$ (where $\operatorname*{C}_{\delta}(f)$\ denotes the
\textquotedblleft$\delta$-approximate certificate complexity\textquotedblright%
\ of $f$).

More dramatically, if Conjecture \ref{folkloreconj}\ holds, then \textit{we
basically cannot hope to prove }$\mathsf{P}\neq\mathsf{BQP}$\textit{ relative
to a random oracle.} \ This would answer a question raised by Fortnow and
Rogers \cite{fr}\ in 1998, and would contrast sharply with the situation for
\textit{non}-random oracles: we have had oracles relative to which
$\mathsf{P}\neq\mathsf{BQP}$, and indeed $\mathsf{BQP}\not \subset
\mathsf{MA}$, since the work of Bernstein and Vazirani \cite{bv}\ in the early
1990s. \ More precisely, under some suitable complexity assumption (such as
$\mathsf{P}=\mathsf{P}^{\mathsf{\#P}}$), we would get $\mathsf{BQP}^{A}%
\subset\mathsf{AvgP}^{A}$\ with probability $1$ for a random oracle $A$.
\ Here $\mathsf{AvgP}$\ is the class of languages for which there exists a
polynomial-time algorithm that solves a\ $1-o\left(  1\right)  $\ fraction of
instances of size $n$. \ In other words, separating $\mathsf{BQP}$ from
$\mathsf{AvgP}$\ relative to a random oracle would be as hard as separating
complexity classes in the unrelativized world. \ This would provide a quantum
counterpart to a theorem of Impagliazzo and Rudich (credited in \cite{kss}),
who used the powerful results of Kahn, Saks, and Smyth \cite{kss} to show that
if $\mathsf{P}=\mathsf{NP}$, then $\mathsf{NP}^{A}\cap\mathsf{coNP}^{A}%
\subset\mathsf{ioAvgP}^{A}$\ with probability $1$ for a random oracle
$A$.\footnote{Here $\mathsf{ioAvgP}$\ means \textquotedblleft average-case
$\mathsf{P}$\ for infinitely many input lengths $n$.\textquotedblright\ \ The
reason Impagliazzo and Rudich only get a simulation in $\mathsf{ioAvgP}$,
rather than $\mathsf{AvgP}$, has to do with the fact that Smyth's result
\cite{smyth} only relates $\operatorname*{D}_{\varepsilon}(f)$\ to
$\operatorname*{C}_{\varepsilon^{3}/30}(f)$, rather than $\operatorname*{D}%
_{\varepsilon+\delta}(f)$\ to $\operatorname*{C}_{\varepsilon}(f)$\ for all
$\delta>0$.}

\subsection{Our Results\label{RESULTS}}

Our main contribution in this paper is essentially to prove Watrous's
conjecture (Conjecture \ref{watrousconj}), that randomized and quantum query
complexities are polynomially related for every symmetric problem.

\begin{theorem}
\label{mainthm}$\operatorname*{R}(f)=O(\operatorname*{Q}(f)^{7}%
\operatorname*{polylog}\operatorname*{Q}(f))$ for every partial function
$f:\left[  M\right]  ^{N}\rightarrow\left\{  0,1,\ast\right\}  $ that is permutation-invariant.
\end{theorem}

We conjecture that $\operatorname*{R}(f)$\ and $\operatorname*{Q}(f)$\ are
polynomially related even for functions $f$ satisfying \textit{one} of the two
symmetries: namely, $f(x_{1},\ldots,x_{N})=f(x_{\sigma\left(  1\right)
},\ldots,x_{\sigma\left(  N\right)  })$ for all $\sigma\in S_{N}$. \ We also
conjecture that the exponent of $7$ can be improved to $2$: in other words,
that Grover's algorithm once again provides the optimal separation between the
quantum and classical models.

While the proof of Theorem \ref{mainthm} is somewhat involved, it can be
entirely understood by those unfamiliar with quantum computing: the
difficulties lie in getting the problem into a form where existing quantum
lower bound technology can be applied to it. \ Let us stress that it was not
at all obvious \textit{a priori} that existing quantum lower bounds would
suffice here; that they did came as a surprise to us.

We first define and analyze a simple randomized algorithm, which tries to
compute $f(X)$\ for a given $X=(x_{1},\ldots,x_{N})$\ by estimating the
multiplicity of each element $x_{i}$. \ Next, by considering where this
randomized algorithm breaks down, we show that one can identify a
\textquotedblleft hard core\textquotedblright\ within $f$: roughly speaking,
two input types $\mathcal{A}^{\ast}$ and $\mathcal{B}^{\ast}$, such that the
difficulty of distinguishing $\mathcal{A}^{\ast}$\ from $\mathcal{B}^{\ast}$
accounts for a polynomial fraction of the entire difficulty of computing $f$.
\ The rest of the proof consists of lower-bounding the quantum query
complexity of distinguishing $\mathcal{A}^{\ast}$\ from $\mathcal{B}^{\ast}$.
\ We do so using a hybrid argument: we develop a \textquotedblleft chopping
procedure\textquotedblright\ that gradually deforms $\mathcal{A}^{\ast}$ to
make it more similar to $\mathcal{B}^{\ast}$, creating a sequence of
intermediate input types $\mathcal{A}_{0}=\mathcal{A}^{\ast},\mathcal{A}%
_{1},\mathcal{A}_{2},\ldots,\mathcal{A}_{2L}=\mathcal{B}^{\ast}$. \ We then
show that, for every $\ell\in\left[  L\right]  $, distinguishing
$\mathcal{A}_{\ell}$\ from $\mathcal{A}_{\ell-1}$\ requires many quantum
queries, \textit{either} by a reduction from Zhandry's recent $\Omega
(N^{1/3})$\ quantum lower bound for the \textsc{SetEquality} problem
\cite{zhandry:col} (which is a nontrivial generalization of Aaronson and Shi's
collision lower bound \cite{as}), or \textit{else} by an application of
Ambainis's general quantum adversary theorem \cite{ambainis}.

Note that, prior to Zhandry's $\Omega(N^{1/3})$\ quantum lower bound for
\textsc{SetEquality}, Midrijanis \cite{midrijanis}\ had proved a lower bound
of $\Omega((N/\log N)^{1/5})$; the latter was the first quantum lower bound
for \textsc{SetEquality}, and the only one for nearly a decade. \ An earlier
version of this paper used Midrijanis's lower bound to show that
$\operatorname*{R}(f)=O(\operatorname*{Q}(f)^{9}\operatorname*{polylog}%
\operatorname*{Q}(f))$ for all permutation-symmetric $f$. \ The improvement to
$\operatorname*{R}(f)=O(\operatorname*{Q}(f)^{7}\operatorname*{polylog}%
\operatorname*{Q}(f))$\ in the current version comes entirely from Zhandry's
improvement of the\ \textsc{SetEquality}\ lower bound to the optimal
$\Omega(N^{1/3})$.

Doing the hybrid argument in the \textquotedblleft obvious\textquotedblright%
\ way produces a bound of the form $\operatorname*{R}(f)\leq\operatorname*{Q}%
(f)^{O\left(  1\right)  }\operatorname*{polylog}N$, which fails to imply a
polynomial relationship between $\operatorname*{R}(f)$\ and $\operatorname*{Q}%
(f)$\ when $\operatorname*{Q}(f)\leq\left(  \log N\right)  ^{o\left(
1\right)  }$. \ However, a more sophisticated hybrid argument eliminates the
$\operatorname*{polylog}N$ factor.

Our second contribution is more exploratory, something we put forward in the
hope of inspiring followup work. \ We study Conjecture \ref{folkloreconj},
which states that every\ $T$-query quantum algorithm can be simulated on
\textit{most} inputs using $T^{O\left(  1\right)  }$\ classical queries. \ We
relate this conjecture to a fundamental open problem in Fourier analysis and
approximation theory.\ Given a real polynomial $p:\mathbb{R}^{N}%
\rightarrow\mathbb{R}$, let%
\[
\operatorname*{Inf}\nolimits_{i}\left[  p\right]  :=\operatorname*{E}%
_{X\in\left\{  0,1\right\}  ^{N}}\left[  (p(X)-p(X^{i}))^{2}\right]
\]
be the \textit{influence} of the $i^{th}$ variable, where $X^{i}$\ means
$X$\ with the $i^{th}$\ bit flipped. \ Then we conjecture that \textit{every
bounded low-degree polynomial has a \textquotedblleft highly
influential\textquotedblright\ variable}. \ More precisely:

\begin{conjecture}
[Bounded Polynomials Have Influential Variables]\label{infconj}Let
$p:\mathbb{R}^{N}\rightarrow\mathbb{R}$\ be a polynomial of degree $d$.
\ Suppose that $0\leq p(X)\leq1$\ for all $X\in\left\{  0,1\right\}  ^{N}$,
and%
\[
\operatorname*{E}_{X\in\left\{  0,1\right\}  ^{N}}\left[
(p(X)-\operatorname*{E}\left[  p\right]  )^{2}\right]  \geq\varepsilon.
\]
Then there exists an $i$ such that $\operatorname*{Inf}_{i}\left[  p\right]
\geq(\varepsilon/d)^{O\left(  1\right)  }$.
\end{conjecture}

We show the following:

\begin{theorem}
\label{consequences}Assume Conjecture \ref{infconj}. \ Then

\begin{enumerate}
\item[(i)] Conjecture \ref{folkloreconj} holds.

\item[(ii)] $\operatorname*{D}_{\varepsilon+\delta}(f)\leq\left(
\operatorname*{Q}_{\varepsilon}(f)/\delta\right)  ^{O\left(  1\right)  }$ for
all Boolean functions $f:\left\{  0,1\right\}  ^{N}\rightarrow\left\{
0,1\right\}  $\ and all $\varepsilon,\delta>0$.

\item[(iii)] If $\mathsf{P}=\mathsf{P}^{\mathsf{\#P}}$, then $\mathsf{BQP}%
^{A}\subset\mathsf{AvgP}^{A}$\ with probability $1$ for a random oracle $A$.
\end{enumerate}
\end{theorem}

The main evidence for Conjecture \ref{infconj}---besides the fact that all the
Fourier analysis experts we asked were confident of it!---is that extremely
similar statements have recently been proved. \ Firstly, O'Donnell, Saks,
Schramm, and Servedio \cite{osss}\ proved an analogue of Conjecture
\ref{infconj}\ for \textit{decision trees}, which are a special case of
bounded real polynomials:

\begin{theorem}
[O'Donnell et al.\ 2005]\label{osssthm}Let $f:\left\{  0,1\right\}
^{N}\rightarrow\left\{  0,1\right\}  $\ be a Boolean function, and suppose
$\Pr\left[  f=1\right]  \Pr\left[  f=0\right]  \geq\varepsilon$. \ Then there
exists an $i$ such that $\operatorname*{Inf}_{i}\left[  f\right]
\geq4\varepsilon/\operatorname*{D}(f)$, where $\operatorname*{D}(f)$\ is the
decision tree complexity of $f$.
\end{theorem}

Unfortunately, Theorem \ref{osssthm}\ does not directly imply anything about
our problem, even though Beals et al.\ \cite{bbcmw}\ showed that
$\operatorname*{D}(f)$\ and $\operatorname*{Q}(f)$\ are polynomially related
for all total Boolean functions $f$. \ The reason is that the acceptance
probability of a quantum algorithm\ need not approximate a total Boolean function.

The second piece of evidence for Conjecture \ref{infconj}\ comes from a
powerful result of Dinur, Friedgut, Kindler, and O'Donnell \cite{dfko},\ which
implies our conjecture, \textit{except} with $\operatorname*{Inf}_{i}\left[
p\right]  \geq\varepsilon^{3}/2^{O\left(  d\right)  }$\ instead of
$\operatorname*{Inf}_{i}\left[  p\right]  \geq\left(  \varepsilon/d\right)
^{O\left(  1\right)  }$. \ Let us state the special case of their result that
is relevant for us:

\begin{theorem}
[Dinur et al.\ 2006]\label{dfkothm0}Let $\varepsilon>0$, and let
$p:\mathbb{R}^{N}\rightarrow\mathbb{R}$\ be a degree-$d$ polynomial such that
$0\leq p(X)\leq1$\ for all $X\in\left\{  0,1\right\}  ^{N}$. \ Then there
exists a $2^{O\left(  d\right)  }/\varepsilon^{2}$-junta $\widetilde{p}%
:\mathbb{R}^{N}\rightarrow\mathbb{R}$ (that is, a polynomial depending on at
most $2^{O\left(  d\right)  }/\varepsilon^{2}$\ variables) such that%
\[
\operatorname*{E}_{X\in\left\{  0,1\right\}  ^{N}}\left[  \left(
\widetilde{p}(X)-p(X)\right)  ^{2}\right]  \leq\varepsilon.
\]

\end{theorem}

Even though Theorem \ref{dfkothm0} has an exponential rather than polynomial
dependence on $1/d$, we observe that it \textit{already} has a nontrivial
consequence for quantum computation. \ Namely, it implies that any $T$-query
quantum algorithm can be simulated on \textit{most} inputs using $2^{O\left(
T\right)  }$\ classical queries.\footnote{Indeed, in this case the classical
queries are nonadaptive.} \ Recall that the gaps between classical and quantum
query complexities can be superexponential (and even $N^{\Omega\left(
1\right)  }$\ versus $O\left(  1\right)  $, as in the example of Aaronson
\cite{aar:ph}), so even an exponential upper bound\ is far from obvious.

\subsection{Subsequent Work\label{SUBSEQ}}

Since the first version of this paper was circulated, there have been at least
three interesting developments (not counting the $\Omega(N^{1/3})$\ quantum
lower bound of Zhandry \cite{zhandry:col}\ for \textsc{SetEquality}, which we
incorporate here).

First, Yuen \cite{yuen:col}\ adapted the hybrid argument that we used to prove
Theorem \ref{mainthm}, in order to show that distinguishing a random function
$X:\left[  N\right]  \rightarrow\left[  N\right]  $\ from a random permutation
requires $\Omega(N^{1/5}/\log N)$\ quantum queries. \ (Subsequently, however,
Zhandry \cite{zhandry:col}\ proved a tight lower bound of\ $\Omega(N^{1/3})$
for the random function versus random permutation problem, using completely
different ideas.)

Second, Montanaro \cite{montanaro:hyper}\ used a hypercontractive inequality
to prove Conjecture \ref{infconj}, in the special case where $p$ is a
multilinear polynomial all of whose coefficients (when written in the Fourier
basis) have the same absolute value. \ Currently, it remains open to
generalize Montanaro's technique to arbitrary multilinear polynomials, let
alone arbitrary polynomials.

Third, Ba\v{c}kurs and Bavarian \cite{bacbav}\ solved a technical problem that
arose from an earlier version of this paper. \ In the earlier version, we
stated Conjecture \ref{infconj}\ in terms of $L_{1}$-influences rather than
$L_{2}$-influences,\ and we also used the $L_{1}$-norm\ in proving the
consequences of Conjecture \ref{infconj}\ for quantum query complexity.
\ Subsequently, Ba\v{c}kurs (personal communication) found an error in our
proof. \ Fortunately, however, we noticed that (a) our proof could be fixed by
simply switching from $L_{1}$-norm\ to $L_{2}$-norm\ throughout, and (b) the
$L_{2}$\ version of Conjecture \ref{infconj}\ was, in any case, provably
equivalent to our original $L_{1}$\ version. \ So we switched to the $L_{2}%
$-norm. \ At the same time, though, we remained curious about our original
$L_{1}$-based argument\ \textit{could have} worked. \ The question boiled down
to the following: given a degree-$d$ real polynomial $p:\mathbb{R}%
^{N}\rightarrow\mathbb{R}$, let%
\[
\operatorname*{Inf}\nolimits_{i}^{1}\left[  p\right]  :=\operatorname*{E}%
_{X\in\left\{  0,1\right\}  ^{N}}\left[  \left\vert p(X)-p(X^{i})\right\vert
\right]  .
\]
Then do we have $\sum_{i=1}^{N}\operatorname*{Inf}\nolimits_{i}^{1}\left[
p\right]  \leq d^{O\left(  1\right)  }$, whenever $p(X)\in\left[  0,1\right]
$\ for all $X\in\left\{  0,1\right\}  ^{N}$? \ Ba\v{c}kurs and Bavarian
\cite{bacbav} show that the answer is yes: indeed, the sum of the $L_{1}%
$-influences is upper-bounded by $O(d^{3}\log d)$. \ Using their result, one
\textit{can} salvage our original $L_{1}$-based argument.

For simplicity, though, in this version of the paper we stick with $L_{2}%
$-influences. \ There, the analogue of Ba\v{c}kurs and Bavarian's result is
much easier, and states that $\sum_{i=1}^{N}\operatorname*{Inf}\nolimits_{i}%
\left[  p\right]  \leq d$\ (we provide the folklore proof in Lemma
\ref{qlbinf}). \ For completeness, in Appendix \ref{NORMS}\ we prove the
equivalence of the $L_{1}$\ and $L_{2}$\ versions\ of Conjecture \ref{infconj}.

\section{Quantum Lower Bound for All Symmetric Problems\label{SYMM}}

In this section we prove Theorem \ref{mainthm}: that $\operatorname*{R}%
(f)=O(\operatorname*{Q}(f)^{7}\operatorname*{polylog}\operatorname*{Q}(f))$
for all permutation-symmetric $f$.

We start with a simple observation that is essential to everything that
follows. \ Since $f$ is symmetric, we can group the inputs $X=(x_{1}%
,\ldots,x_{N})$ into equivalence classes that we call \textit{types}.

\begin{definition}
Given an input $X=(x_{1},\ldots,x_{N})\in\left[  M\right]  ^{N}$, the type of
$X$ is a list of positive integers $\mathcal{A}=(a_{1},\ldots,a_{u})$, which
records the multiplicities of the integers occurring in $X$ from most to least
frequent. \ So in particular, $a_{1}\geq\cdots\geq a_{u}$\ and $a_{1}%
+\cdots+a_{u}=N$. \ For convenience, we adopt the convention that $a_{i}%
=0$\ for all $i>u$.
\end{definition}

In other words, a type is just a partition (or \textit{Young diagram})\ that
records the multiplicities of the input elements. \ For example, a one-to-one
input has type $a_{1}=\cdots=a_{N}=1$, while a two-to-one input has type
$a_{1}=\cdots=a_{N/2}=2$. \ We write $X\in\mathcal{A}$\ if $X$ is of type
$\mathcal{A}$. \ Clearly $f(X)$\ depends only on the type of $X$.
\ Furthermore, given a quantum query algorithm $Q$, we can assume without loss
of generality that $\Pr\left[  Q~\text{accepts }X\right]  $\ depends only on
the type of $X$---since we can \textquotedblleft symmetrize\textquotedblright%
\ $Q$ (that is, randomly permute $X$'s inputs and outputs)\ prior to running
$Q$.

\subsection{Randomized Upper Bound\label{RUB}}

Let $X=(x_{1},\ldots,x_{N})$\ be an input. \ For each $j\in\left[  M\right]
$, let $\kappa_{j}$\ be the number of $i$'s such that $x_{i}=j$. \ Then the
first step is to give a classical randomized algorithm that estimates the
$\kappa_{j}$'s. \ This algorithm, $\mathcal{S}_{T}$, is an extremely
straightforward sampling procedure. \ (Indeed, there is essentially nothing
else that a randomized algorithm can do here.) \ $\mathcal{S}_{T}$ will make
$O(T^{1+c}\log T)$ queries, where $T$ is a parameter and $c\in\left(
0,1\right]  $ is a constant that we will choose later to optimize the final
bound.\bigskip

\texttt{\qquad Set }$U:=21T^{1+c}\ln T$

\texttt{\qquad Choose }$U$\texttt{\ indices }$i_{1},\ldots,i_{U}\in\left[
N\right]  $\texttt{\ uniformly at random with replacement}

\texttt{\qquad Query }$x_{i_{1}},\ldots,x_{i_{U}}$

\texttt{\qquad For each }$j\in\left[  M\right]  $\texttt{:}

\texttt{\qquad\qquad Let }$z_{j}$\texttt{\ be the number of occurrences of
}$j$\texttt{ in }$(x_{i_{1}},\ldots,x_{i_{U}})$

\texttt{\qquad\qquad Output }$\widetilde{\kappa}_{j}:=\frac{z_{j}}{U}\cdot
N$\texttt{ as the estimate for }$\kappa_{j}$ \bigskip

We now analyze how well $\mathcal{S}_{T}$\ works.

\begin{lemma}
\label{sampling}With probability $1-O\left(  1/T\right)  $,\ we have
$\left\vert \widetilde{\kappa}_{j}-\kappa_{j}\right\vert \leq\frac{N}{T}%
+\frac{\kappa_{j}}{T^{c}}$\ for all $j\in\left[  M\right]  $.
\end{lemma}

\begin{proof}
For each $j\in\left[  M\right]  $, we consider four cases. \ First suppose
$\kappa_{j}\geq N/T^{1-c}$. \ Notice that $z_{j}$\ is a sum of $U$%
\ independent Boolean variables, and that $\operatorname*{E}\left[
z_{j}\right]  =\frac{U}{N}\operatorname*{E}[\widetilde{\kappa}_{j}]=\frac
{U}{N}\kappa_{j}$. \ Thus%
\begin{align*}
\Pr\left[  \left\vert \widetilde{\kappa}_{j}-\kappa_{j}\right\vert
>\frac{\kappa_{j}}{T^{c}}\right]   &  =\Pr\left[  \left\vert z_{j}-\frac{U}%
{N}\kappa_{j}\right\vert >\frac{U\kappa_{j}}{NT^{c}}\right] \\
&  <2\exp\left(  -\frac{U\kappa_{j}/N}{3T^{2c}}\right) \\
&  <2\exp\left(  -\frac{U}{3T^{1+c}}\right) \\
&  =2T^{-7},
\end{align*}
where the second line follows from a Chernoff bound and the third from
$\kappa_{j}\geq N/T^{1-c}$.

Second, suppose $N/T\leq\kappa_{j}<N/T^{1-c}$. \ Then%
\begin{align*}
\Pr\left[  \left\vert \widetilde{\kappa}_{j}-\kappa_{j}\right\vert >\frac
{N}{T}\right]   &  =\Pr\left[  \left\vert z_{j}-\frac{U}{N}\kappa
_{j}\right\vert >\frac{U}{T}\right] \\
&  <2\exp\left(  -\frac{U\kappa_{j}/N}{3}\left(  \frac{N}{T\kappa_{j}}\right)
^{2}\right) \\
&  <2\exp\left(  -\frac{U}{3T^{1+c}}\right) \\
&  =2T^{-7}%
\end{align*}
where the second line follows from a Chernoff bound (which is valid because
$\frac{N}{T\kappa_{j}}\leq1$) and the third from $\kappa_{j}<N/T^{1-c}$.

Third, suppose $N/T^{6}\leq\kappa_{j}<N/T$. \ Then%
\begin{align*}
\Pr\left[  \left\vert \widetilde{\kappa}_{j}-\kappa_{j}\right\vert >\frac
{N}{T}\right]   &  =\Pr\left[  \left\vert z_{j}-\frac{U}{N}\kappa
_{j}\right\vert >\frac{U}{T}\right] \\
&  <\left(  \frac{e^{N/\left(  T\kappa_{j}\right)  }}{\left(  1+N/\left(
T\kappa_{j}\right)  \right)  ^{1+N/\left(  T\kappa_{j}\right)  }}\right)
^{U\kappa_{j}/N}\\
&  \leq\exp\left(  -\frac{N}{T\kappa_{j}}\cdot\frac{U\kappa_{j}}{N}\right) \\
&  =\exp\left(  -\frac{U}{T}\right) \\
&  =O\left(  \frac{1}{T^{7}}\right)  ,
\end{align*}
where the second line follows from a Chernoff bound, the third line follows
from $\frac{N}{T\kappa_{j}}>1$, and the last follows from $U=21T^{1+c}\ln T$.

Fourth, suppose $\kappa_{j}<N/T^{6}$. \ Then%
\begin{align*}
\Pr\left[  \left\vert \widetilde{\kappa}_{j}-\kappa_{j}\right\vert >\frac
{N}{T}\right]   &  =\Pr\left[  \left\vert z_{j}-\frac{U}{N}\kappa
_{j}\right\vert >\frac{U}{T}\right] \\
&  \leq\Pr\left[  z_{j}\geq2\right] \\
&  \leq\binom{U}{2}\left(  \frac{\kappa_{j}}{N}\right)  ^{2}\\
&  \leq\frac{U^{2}}{T^{6}}\left(  \frac{\kappa_{j}}{N}\right) \\
&  \leq\frac{\kappa_{j}}{TN}%
\end{align*}
for all sufficiently large $T$, where the second line follows from $\kappa
_{j}<N/T^{6}$, the third from the union bound, the fourth from $\kappa
_{j}<N/T^{6}$ (again), and the fifth from $U\leq21T^{2}\ln T$.

Notice that there are at most $T^{6}$\ values of $j$ such that $\kappa_{j}\geq
N/T^{6}$. \ Hence, putting all four cases together,%
\begin{align*}
\Pr\left[  \exists j:\left\vert \widetilde{\kappa}_{j}-\kappa_{j}\right\vert
>\frac{N}{T}+\frac{\kappa_{j}}{T^{c}}\right]   &  \leq T^{6}\cdot O\left(
\frac{1}{T^{7}}\right)  +\sum_{j:\kappa_{j}<N/T^{6}}\frac{\kappa_{j}}{TN}\\
&  =O\left(  \frac{1}{T}\right)  .
\end{align*}

\end{proof}

Now call $\mathcal{A}$ a $1$\textit{-type} if $f(X)=1$\ for all $X\in
\mathcal{A}$, or a $0$\textit{-type} if $f(X)=0$\ for all $X\in\mathcal{A}$.
\ Consider the following randomized algorithm $\mathcal{R}_{T}$\ to compute
$f(X)$:\bigskip

\qquad\texttt{Run }$\mathcal{S}_{T}$\texttt{ to find an estimate
}$\widetilde{\kappa}_{i}$\texttt{\ for each }$\kappa_{i}$

\qquad\texttt{Sort the }$\widetilde{\kappa}_{i}$\texttt{'s in descending
order, so that }$\widetilde{\kappa}_{1}\geq\cdots\geq\widetilde{\kappa}_{M}$

\qquad\texttt{If there exists a }$1$\texttt{-type\ }$\mathcal{A}=(a_{1}%
,a_{2},\ldots)$\texttt{\ such that }$\left\vert \widetilde{\kappa}_{i}%
-a_{i}\right\vert \leq\frac{N}{T}+\frac{a_{i}}{T^{c}}$

\qquad\qquad\texttt{for all }$i$\texttt{, then output }$f(X)=1$

\qquad\texttt{Otherwise output }$f(X)=0\bigskip$

Clearly $\mathcal{R}_{T}$\ makes $O(T^{1+c}\log T)$\ queries, just as
$\mathcal{S}_{T}$\ does. \ We now give a sufficient condition for
$\mathcal{R}_{T}$\ to succeed.

\begin{lemma}
\label{rworks}Suppose that for all $1$-types $\mathcal{A}=(a_{1},a_{2}%
,\ldots)$ and $0$-types $\mathcal{B}=(b_{1},b_{2},\ldots)$, there exists an
$i$\ such that $\left\vert a_{i}-b_{i}\right\vert >\frac{2N}{T}+\frac
{a_{i}+b_{i}}{T^{c}}$. \ Then $\mathcal{R}_{T}$\ computes $f$\ with bounded
probability of error, and hence $\operatorname*{R}\left(  f\right)
=O(T^{1+c}\log T)$.
\end{lemma}

\begin{proof}
First suppose $X\in\mathcal{A}$\ where $\mathcal{A}=(a_{1},a_{2},\ldots)$\ is
a $1$-type. \ Then by Lemma \ref{sampling}, with probability $1-O\left(
1/T\right)  $\ we have $\left\vert \widetilde{\kappa}_{i}-a_{i}\right\vert
\leq\frac{N}{T}+\frac{a_{i}}{T^{c}}$\ for all $i$. \ (It is easy to see that
sorting the $\widetilde{\kappa}_{i}$'s can only decrease the maximum
difference.) \ Provided this occurs, $\mathcal{R}_{T}$\ finds some
$1$-type\ close to $(\widetilde{\kappa}_{1},\widetilde{\kappa}_{2},\ldots)$
(possibly $\mathcal{A}$\ itself) and outputs $f(X)=1$.

Second, suppose $X\in\mathcal{B}$\ where $\mathcal{B}=(b_{1},b_{2},\ldots
)$\ is a $0$-type. \ Then with probability $1-O\left(  1/T\right)  $ we have
$\left\vert \widetilde{\kappa}_{i}-b_{i}\right\vert \leq\frac{N}{T}%
+\frac{b_{i}}{T^{c}}$\ for all $i$. \ Provided this occurs, by the triangle
inequality, for every $1$-type\ $\mathcal{A}=(a_{1},a_{2},\ldots)$\ there
exists an $i$ such that%
\[
\left\vert \widetilde{\kappa}_{i}-a_{i}\right\vert \geq\left\vert a_{i}%
-b_{i}\right\vert -\left\vert \widetilde{\kappa}_{i}-b_{i}\right\vert
>\frac{N}{T}+\frac{a_{i}}{T^{c}}.
\]
Hence $\mathcal{R}_{T}$\ does \textit{not} find a $1$-type\ close to
$(\widetilde{\kappa}_{1},\widetilde{\kappa}_{2},\ldots)$, and it outputs
$f(X)=0$.
\end{proof}

In particular, suppose we keep decreasing $T$ until there exists a
$1$-type\ $\mathcal{A}^{\ast}=(a_{1},a_{2},\ldots)$ and a $0$%
-type\ $\mathcal{B}^{\ast}=(b_{1},b_{2},\ldots)$ such that%
\begin{equation}
\left\vert a_{i}-b_{i}\right\vert \leq\frac{2N}{T}+\frac{a_{i}+b_{i}}{T^{c}}
\label{abbound}%
\end{equation}
for all $i$, stopping as soon as that happens. \ Then Lemma \ref{rworks}
implies that we will still have $\operatorname*{R}\left(  f\right)
=O(T^{1+c}\log T)$. \ For the rest of the proof, we will fix that
\textquotedblleft almost as small as possible\textquotedblright\ value of $T$
for which (\ref{abbound}) holds, as well as the $1$-type\ $\mathcal{A}^{\ast}%
$\ and the $0$-type\ $\mathcal{B}^{\ast}$ that $\mathcal{R}_{T}$%
\ \textquotedblleft just barely distinguishes\textquotedblright\ from one another.

\subsection{The Chopping Procedure\label{CHOP}}

Given two sets of inputs $A$ and $B$ with $A\cap B=\varnothing$,\ let
$\operatorname*{Q}(A,B)$\ be the minimum number of queries made by any quantum
algorithm that accepts every $X\in A$\ with probability at least $2/3$, and
accepts every $Y\in B$ with probability at most $1/3$. \ Also, let
$\operatorname*{Q}_{\varepsilon}(A,B)$\ be the minimum number of queries made
by any quantum algorithm that accepts every $X\in A$\ with at least some
probability $p$, and that accepts every $Y\in B$ with probability at most
$p-\varepsilon$. \ Then we have the following basic relation:

\begin{proposition}
\label{amplify}$\operatorname*{Q}(A,B)=O(\frac{1}{\varepsilon}%
\operatorname*{Q}_{\varepsilon}(A,B))$ for all $A,B$\ and all $\varepsilon>0$.
\end{proposition}

\begin{proof}
This follows from standard amplitude estimation techniques (see Brassard et
al.\ \cite{bhmt}\ for example).
\end{proof}

The rest of the proof consists of lower-bounding $\operatorname*{Q}%
(\mathcal{A}^{\ast},\mathcal{B}^{\ast})$, the quantum query complexity of
distinguishing inputs of type $\mathcal{A}^{\ast}$ from inputs of type
$\mathcal{B}^{\ast}$. \ We do this via a hybrid argument. \ Let
$L:=\left\lceil \log_{2}N\right\rceil +1$. \ At a high level, we will
construct a sequence of types $\mathcal{A}_{0},\ldots,\mathcal{A}_{2L}$
%\ and
%$\mathcal{B}_{0},\ldots,\mathcal{B}_{L}$,
such that

\begin{enumerate}
\item[(i)] $\mathcal{A}_{0}=\mathcal{A}^{\ast}$,

\item[(ii)] $\mathcal{A}_{2L}=\mathcal{B}^{\ast}$, and

\item[(iii)] $\operatorname*{Q}(\mathcal{A}_{\ell},\mathcal{A}_{\ell-1})$ is
large for every $\ell\in\left[  2L\right]  $.
\end{enumerate}

Provided we can do this, it is not hard to see that we get the desired lower
bound on $\operatorname*{Q}(\mathcal{A}^{\ast},\mathcal{B}^{\ast})$. \ Suppose
a quantum algorithm distinguishes $\mathcal{A}_{0}=\mathcal{A}^{\ast}$\ from
$\mathcal{A}_{2L}=\mathcal{B}^{\ast}$ with constant bias. \ Then by the
triangle inequality, it must also distinguish \textit{some} $\mathcal{A}%
_{\ell}$\ from $\mathcal{A}_{\ell+1}$ with reasonably large bias (say
$\Omega\left(  1/\log N\right)  $). \ By Proposition \ref{amplify}, any
quantum algorithm that succeeds with\ bias $\varepsilon$\ can be amplified,
with $O\left(  1/\varepsilon\right)  $ overhead, to an algorithm that succeeds
with constant bias.

Incidentally, the need, in this hybrid argument, to amplify the distinguishing
bias $\varepsilon=\varepsilon_{\ell}$\ from $\Omega\left(  1/\log N\right)  $
to $\Omega\left(  1\right)  $ is exactly what could produce an undesired
$1/\log N$\ factor in our final lower bound on $\operatorname*{Q}(f)$, if we
were not careful. \ (We mentioned this issue in Section \ref{RESULTS}.) \ The
way we will solve this problem, roughly speaking, is to design the
$\mathcal{A}_{\ell}$'s in such a way that our lower bounds on
$\operatorname*{Q}(\mathcal{A}_{\ell},\mathcal{A}_{\ell-1})$ increase quickly
as functions of $\ell$. \ That way, we can take the biases $\varepsilon_{\ell
}$ to\ decrease quadratically with $\ell$ (thus summing to a constant), yet
still have $\operatorname*{Q}(\mathcal{A}_{\ell},\mathcal{A}_{\ell-1})$
increasing quickly enough that%
\begin{align*}
\operatorname*{Q}\nolimits_{\varepsilon_{\ell}}(\mathcal{A}_{\ell}%
,\mathcal{A}_{\ell-1})  &  =\Omega(\varepsilon_{\ell}\operatorname*{Q}%
(\mathcal{A}_{\ell},\mathcal{A}_{\ell-1}))
\end{align*}
remain \textquotedblleft uniformly large,\textquotedblright\ with $1/\log
T$\ factors but no $1/\log N$\ factor.

We now describe the procedure for creating the intermediate types
$\mathcal{A}_{\ell}$.\ \ Intuitively, we want to form $\mathcal{A}_{\ell}%
$\ from $\mathcal{A}_{\ell-1}$ by making its Young diagram more similar to
that of $\mathcal{B}^{\ast}$, by decreasing the rows of $\mathcal{A}_{\ell-1}$
which are larger than the corresponding rows of $\mathcal{B}^{\ast}$ and
increasing the rows of $\mathcal{A}_{\ell-1}$ which are smaller than the
corresponding rows of $\mathcal{B}^{\ast}$.

More precisely, we construct the intermediate types $\mathcal{A}%
_{1},\mathcal{A}_{2},\ldots$\ via the following procedure. \ In this
procedure, $(a_{1}, a_{2}, \ldots)$\ is an input type that is initialized to
$\mathcal{A}^{\ast}$, and $\mathcal{B}^{\ast}=(b_{1},b_{2},\ldots)$.\bigskip

\qquad\texttt{let }$P$\texttt{ be the first power of }$2$\texttt{ greater than
or equal to }$N$

\qquad\texttt{for }$\ell:=1$\texttt{\ to }$L$

\qquad\qquad\texttt{let }$S_{A}$\texttt{ be the set of }$i$\texttt{ such that
} $a_{i}-b_{i}\geq P/2^{l}$

\qquad\qquad\texttt{let }$S_{B}$\texttt{ be the set of }$i$\texttt{ such that
} $b_{i}-a_{i}\geq P/2^{l}$

\qquad\qquad\texttt{let }$m:=\min(\left\vert S_{A}\right\vert ,\left\vert
S_{B}\right\vert )$

\qquad\qquad\texttt{choose }$m$\texttt{ elements }$i$\texttt{ from }$S_{A}%
$\texttt{, set } $a_{i}:=a_{i}-P/2^{\ell}$\texttt{ and remove them from
}$S_{A}$

\qquad\qquad\texttt{choose }$m$\texttt{ elements }$i$\texttt{ from }$S_{B}%
$\texttt{, set } $a_{i}:=a_{i}+P/2^{\ell}$\texttt{ and remove them from
}$S_{B}$

\qquad\qquad\texttt{let }$\mathcal{A}_{2\ell-1}:=\operatorname*{type}%
(a_{1},a_{2},\ldots)$

\qquad\qquad\texttt{if }$\left\vert S_{A}\right\vert >0$

\qquad\qquad\qquad\texttt{let } $a_{i}:=a_{i}-P/2^{\ell}$ \texttt{ for all }
$i\in S_{A}$

\qquad\qquad\qquad\texttt{choose }$\left\vert S_{A}\right\vert $\texttt{
elements }$i$\texttt{ such that }$a_{i}<b_{i}$ \texttt{ and set }
$a_{i}:=a_{i}+P/2^{\ell}$

\qquad\qquad\texttt{if }$\left\vert S_{B}\right\vert >0$

\qquad\qquad\qquad\texttt{let } $a_{i}:=a_{i}+P/2^{\ell}$ \texttt{ for all }
$i\in S_{B}$

\qquad\qquad\qquad\texttt{choose }$\left\vert S_{B}\right\vert $\texttt{
elements }$i$\texttt{ such that }$a_{i}>b_{i}$ \texttt{ and set }
$a_{i}:=a_{i}-P/2^{\ell}$

\qquad\qquad\texttt{let }$\mathcal{A}_{2\ell}:=\operatorname*{type}%
(a_{1},a_{2},\ldots)$

\qquad\texttt{next }$\ell\bigskip$

The procedure is illustrated pictorially in Figure \ref{chopfig}.%
%TCIMACRO{\FRAME{ftbpFU}{5.0972in}{1.4702in}{0pt}{\Qcb{Chopping a row of
%$\QTR{cal}{A}_{\ell}$'s Young diagram to make it more similar to
%$\QTR{cal}{B}_{\ell}$.}}{\Qlb{chopfig}}{chopfig.eps}%
%width 5.0972in;  height 1.4702in;  depth 0pt;  original-width 10.6017in;
%original-height 8.0782in;  cropleft "0.1893";  croptop "0.9688";
%cropright "0.7870";  cropbottom "0.7455";
%filename 'chopfig.eps';file-properties "XNPEU";}} }%
%BeginExpansion
\begin{figure}[ptb]%
\centering
\includegraphics[
trim=2.006902in 6.022299in 2.258162in 0.252040in,
height=1.4702in,
width=5.0972in
]%
{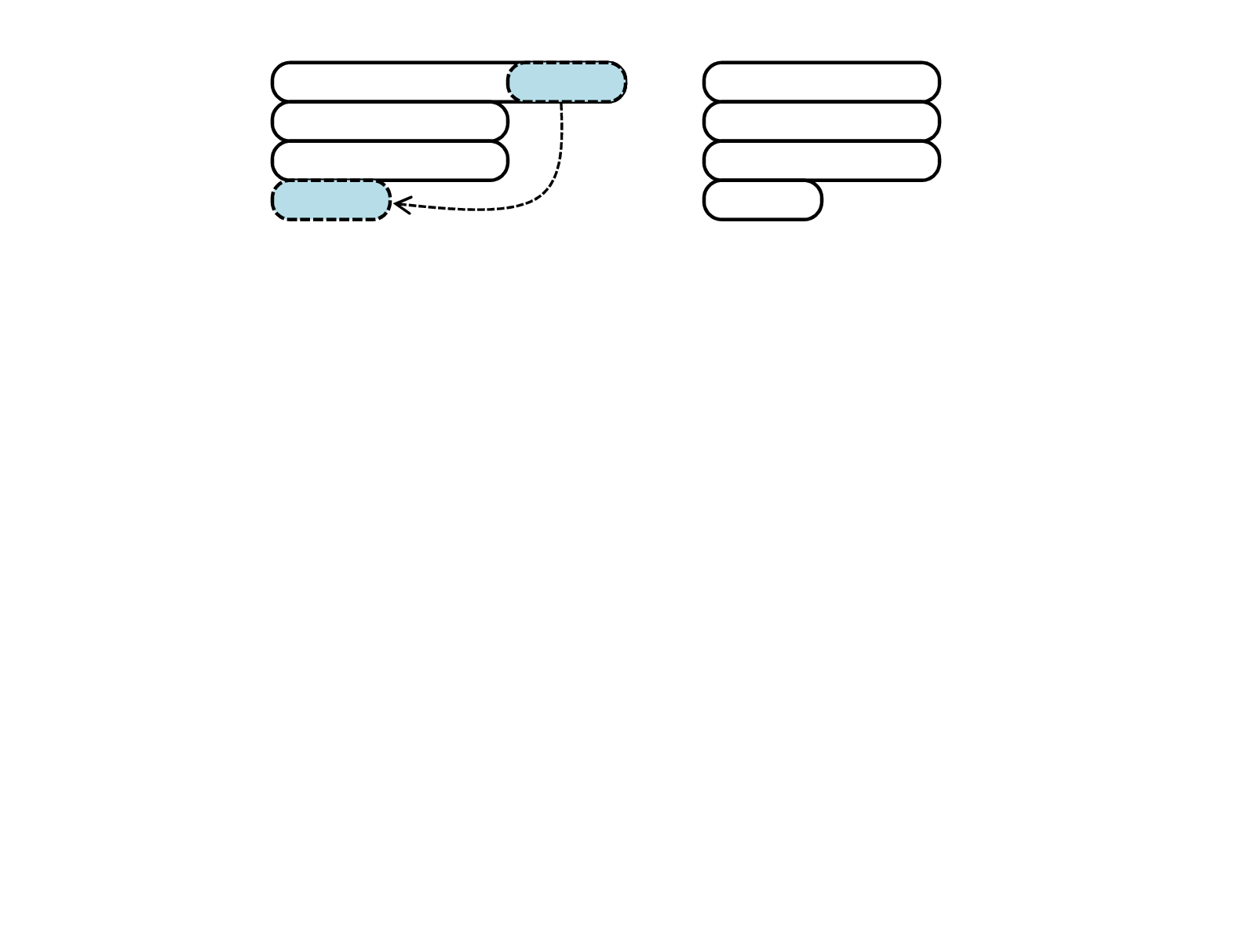}%
\caption{Chopping a row of $\mathcal{A}_{\ell}$'s Young diagram to make it
more similar to $\mathcal{B}_{\ell}$.}%
\label{chopfig}%
\end{figure}
%EndExpansion

We start with some simple observations. \ First, by construction, this
procedure halts after $2L=O\left(  \log N\right)  $ iterations. \ Second,
after the $\ell^{\mathrm{th}}$ iteration, we have $\left\vert a_{i}%
-b_{i}\right\vert <\frac{P}{2^{\ell}}$ for all $i$. \ This follows by
induction. \ Let $a_{i}^{\prime}$ be the value of $a_{i}$ before the
$\ell^{\mathrm{th}}$ iteration. \ Because of the inductive assumption, we must
have $\left\vert a_{i}^{\prime}-b_{i}\right\vert <\frac{P}{2^{\ell-1}}$---for
if $\left\vert a_{i}^{\prime}-b_{i}\right\vert \geq\frac{P}{2^{\ell}}$, then
$a_{i}$ is changed by $\frac{P}{2^{\ell}}$ during the $\ell^{\mathrm{th}}$
iteration, to decrease the difference $\left\vert a_{i}-b_{i}\right\vert $.
\ After this change,
\[
\left\vert a_{i}-b_{i}\right\vert =\left\vert a_{i}^{\prime}-b_{i}\right\vert
-\frac{P}{2^{\ell}}<\frac{P}{2^{\ell-1}}-\frac{P}{2^{\ell}}=\frac{P}{2^{\ell}%
}.
\]
Besides the $\left\vert a_{i}^{\prime}-b_{i}\right\vert \geq\frac{P}{2^{\ell}%
}$ case, there is one other case where $\left\vert a_{i}-b_{i}\right\vert $
could change. \ In the transition from $A_{2\ell-1}$ to $A_{2\ell}$, if
$\left\vert S_{A}\right\vert >0$ or $\left\vert S_{B}\right\vert >0$, then we
change $a_{i}$ for $\left\vert S_{A}\right\vert $ or $\left\vert
S_{B}\right\vert $ elements $i$ that do not belong to $S_{A}$ or $S_{B}$.
\ For those elements, we have $\left\vert a_{i}-b_{i}\right\vert <\frac
{P}{2^{\ell}}$ and we change $a_{i}$ in the direction of $b_{i}$ (we increase
it by $\frac{P}{2^{\ell}}$ if $a_{i}<b_{i}$ and decrease it by the same amount
if $a_{i}>b_{i}$). \ Therefore, after the change, the sign of the difference
$a_{i}-b_{i}$ flips and $\left\vert a_{i}-b_{i}\right\vert <\frac{P}{2^{\ell}%
}$.

Now let us define%
\[
\left\Vert \mathcal{A}-\mathcal{B}\right\Vert :=\frac{1}{2}\sum_{i=1}%
^{N}\left\vert a_{i}-b_{i}\right\vert .
\]
Notice that $\left\Vert \mathcal{A}_{\ell}-\mathcal{A}_{\ell-1}\right\Vert
=rP/2^{\ell^{\prime}}$, where $r$\ is the number of rows that get increased
(or decreased) in the $\ell^{th}$\ iteration and $l^{\prime}=\lceil\frac{l}%
{2}\rceil$. \ We now prove an upper bound on $\left\Vert \mathcal{A}_{\ell
}-\mathcal{A}_{\ell-1}\right\Vert $\ when $\ell$ is small, which will be
useful later.

\begin{lemma}
\label{chopub}If $\ell\leq\left(  \log_{2}T\right)  -2$, then%
\[
\left\Vert \mathcal{A}_{2\ell-2}-\mathcal{A}_{2\ell-1}\right\Vert +\left\Vert
\mathcal{A}_{2\ell-1}-\mathcal{A}_{2\ell}\right\Vert \leq\frac{4N}{T^{c}}.
\]

\end{lemma}

\begin{proof}
Let $m:=\max(\left\vert S_{A}\right\vert ,\left\vert S_{B}\right\vert )$. Then%
\[
\left\Vert \mathcal{A}_{2\ell-2}-\mathcal{A}_{2\ell-1}\right\Vert +\left\Vert
\mathcal{A}_{2\ell-1}-\mathcal{A}_{2\ell}\right\Vert =m\frac{P}{2^{\ell}}.
\]
Without loss of generality, we assume that $m=\left\vert S_{A}\right\vert $.
\ To show the lemma, it suffices to prove that $\left\vert S_{A}\right\vert
\leq\frac{4N/T^{c}}{P/2^{\ell}}$.

We consider the sum $\sum_{j\in R}\left\vert a_{j}-b_{j}\right\vert $ where
$R$ is the set of all $j$ such that $\left\vert a_{j}-b_{j}\right\vert
\geq\frac{P}{2^{\ell}}$, with $(a_{1},a_{2},\ldots)$ evolving from
$\mathcal{A}_{0}$ to $\mathcal{A}_{2\ell-2}$ and $\mathcal{B}^{\ast}%
=(b_{1},b_{2},\ldots)$ fixed. \ Initially (when $(a_{1},a_{2},\ldots
)=\mathcal{A}_{0}$), we have
\[
\frac{P}{2^{\ell}}\leq\left\vert a_{j}-b_{j}\right\vert \leq\frac{2N}{T}%
+\frac{a_{j}+b_{j}}{T^{c}}%
\]
for each $j\in R$. \ Since $\ell\leq\left(  \log_{2}T\right)  -2$, the left
inequality implies%
\[
\left\vert a_{j}-b_{j}\right\vert \geq\frac{4N}{T},
\]
which combined with the right inequality yields%
\begin{equation}
\frac{a_{j}+b_{j}}{T^{c}}\geq\frac{2N}{T}. \label{dude}%
\end{equation}
Therefore%
\begin{align*}
\sum_{i\in R}\left\vert a_{i}-b_{i}\right\vert  &  \leq\sum_{i\in R}\left(
\frac{2N}{T}+\frac{a_{i}+b_{i}}{T^{c}}\right) \\
&  \leq2\sum_{i\in R}\frac{a_{i}+b_{i}}{T^{c}}\\
&  \leq\frac{4N}{T^{c}},
\end{align*}
where the third line uses (\ref{dude}).

The sum $\sum_{i\in R}\left\vert a_{i}-b_{i}\right\vert $ is not increased by
any step of the algorithm that generates $\mathcal{A}_{0},\ldots
,\mathcal{A}_{2\ell-2}$. Therefore, at the beginning of the ${\ell
}^{\mathrm{th}}$ iteration, we still have $\sum_{i\in R}\left\vert a_{i}%
-b_{i}\right\vert \leq\frac{4N}{T^{c}}$. This means that $\left\vert
S_{A}\right\vert \leq\frac{4N/T^{c}}{P/2^{\ell}}$.
\end{proof}

\subsection{Quantum Lower Bounds\label{QLB}}

Recall that we listed four properties that we needed the chopping procedure to
satisfy. \ We have already seen that it satisfies properties (i)-(ii),\ so the
remaining step is to show that it satisfies property (iii). \ That is, we need
to lower-bound $\operatorname*{Q}(\mathcal{A}_{\ell},\mathcal{A}_{\ell-1}%
)$,\ the bounded-error quantum query complexity of distinguishing inputs of
type $\mathcal{A}_{\ell}$\ from inputs of type $\mathcal{A}_{\ell-1}$. \ To do
this, it will be convenient to consider two cases: first, that forming
$\mathcal{A}_{\ell}$\ involved chopping few\ elements of $\mathcal{A}_{\ell
-1}$, and second, that it involved chopping many\ elements. \ We will show
that we \textquotedblleft win either way,\textquotedblright\ by a different
quantum lower bound in each case.

First consider the case that few elements were chopped. \ Here we prove a
lower bound using Ambainis's quantum adversary method \cite{ambainis}, in its
\textquotedblleft general\textquotedblright\ form (the one used, for example,
to lower-bound the quantum query complexity of inverting a permutation). \ For
completeness, we now state Ambainis's adversary theorem in the form we will need.

\begin{theorem}
[Ambainis \cite{ambainis}]\label{ambthm}Let $A,B\subseteq\left[  M\right]
^{N}$\ be two sets of inputs with $A\cap B=\varnothing$. \ Let $R\subseteq
A\times B$\ be a relation on input pairs, such that for every $X\in A$\ there
exists at least one $Y\in B$\ with $\left(  X,Y\right)  \in R$\ and vice
versa. \ Given inputs $X=(x_{1},\ldots,x_{N})$\ in $A$ and $Y=(y_{1}%
,\ldots,y_{N})$\ in $B$, let%
\begin{align*}
q_{X,i}  &  =\Pr_{Y\in B}\left[  x_{i}\neq y_{i}~|~\left(  X,Y\right)  \in
R\right]  ,\\
q_{Y,i}  &  =\Pr_{X\in A}\left[  x_{i}\neq y_{i}~|~\left(  X,Y\right)  \in
R\right]  .
\end{align*}
Suppose that $q_{X,i}q_{y,i}\leq\alpha$\ for every $\left(  X,Y\right)  \in
R$\ and every $i\in\left[  N\right]  $\ such that $x_{i}\neq y_{i}$. \ Then
$\operatorname*{Q}(A,B)=\Omega(1/\sqrt{\alpha})$.
\end{theorem}

Using Theorem \ref{ambthm},\ we can prove the following lower bound on
$\operatorname*{Q}\left(  \mathcal{A}_{\ell},\mathcal{A}_{\ell-1}\right)  $.

\begin{lemma}
\label{usingamb}Let $d=\left\Vert \mathcal{A}_{\ell}-\mathcal{A}_{\ell
-1}\right\Vert $, and assume $d\leq N/2$.\ \ Then $\operatorname*{Q}%
(\mathcal{A}_{\ell},\mathcal{A}_{\ell-1})=\Omega(\sqrt{N/d})$.
\end{lemma}

\begin{proof}
Let $\mathcal{A}_{\ell-1}=(a_{1},a_{2},\ldots)$, and\ let $\ell^{\prime
}=\lceil\frac{\ell}{2}\rceil$. \ Then in the transition from $\mathcal{A}%
_{\ell-1}$ to $\mathcal{A}_{\ell}$, we augment or chop various rows by
$P/2^{\ell^{\prime}}$ elements each. \ Let $i\left(  1\right)  ,\ldots
,i\left(  r\right)  $\ be the $r$ rows in $\mathcal{A}_{\ell-1}$\ that get
chopped and let $i^{\prime}\left(  1\right)  ,\ldots,i^{\prime}\left(
r\right)  $\ be the $r$ rows in $\mathcal{A}_{\ell-1}$\ that get augmented.

Fix distinct $h_{1},\ldots,h_{r}\in\left[  M\right]  $ and $h_{1}^{\prime
},\ldots,h_{r}^{\prime}\in\left[  M\right]  $. \ Also, let us restrict
ourselves to inputs\ such that for each $j\in\left[  r\right]  $, there are
exactly $a_{i\left(  j\right)  }$\ indices $i\in\left[  N\right]
$\ satisfying $x_{i}=h_{j}$ and exactly $a_{i^{\prime}\left(  j\right)  }%
$\ indices $i\in\left[  N\right]  $\ satisfying $x_{i}=h_{j}^{\prime}$. G iven
inputs $X=(x_{1},\ldots,x_{N})$ in $\mathcal{A}_{\ell-1}$\ and $Y=(y_{1}%
,\ldots,y_{N})$ in $\mathcal{A}_{\ell}$, we set $\left(  X,Y\right)  \in
R$\ if and only if it is possible to transform $X$ to $Y$ in the following way:

\begin{enumerate}
\item[(1)] For each $j\in\left[  r\right]  $, change exactly $P/2^{\ell
^{\prime}}$\ of the $x_{i}$'s that are equal to $h_{j}$\ to value
$h_{j}^{\prime}$. \ (The total number of changed elements is $d$.)

\item[(2)] Swap the $d$\ elements of $X$\ that were changed in step (2) with
any other $d$\ elements $x_{i}$ of $X$, subject to the following constraints:

\begin{enumerate}
\item[(a)] we do not use $x_{i}$ such that $x_{i}=h_{j}$ for some $j$ and
$\frac{a_{i_{j}}-P/2^{\ell^{\prime}}}{P/2^{\ell^{\prime}}} < \frac{N-d}{3d}$;

\item[(b)] we do not use $x_{i}$ such that $x_{i}=h_{j}^{\prime}$ for some $j$
and $\frac{a_{i_{j}}}{P/2^{\ell^{\prime}}}<\frac{N-d}{3d}$.
\end{enumerate}
\end{enumerate}

The procedure is illustrated pictorially in Figure \ref{swapfig}.%
%TCIMACRO{\FRAME{ftbpFU}{3.7343in}{1.1398in}{0pt}{\Qcb{In this example, $N=11$,
%$r=2$, $P/2^{\ell}=2$, and $a_{1}=a_{2}=3$. \ So we transform $X$ to $Y$ by
%choosing $h_{1}=1$\ and $h_{2}=2$, changing any two elements equal to $h_{1}$
%and any two elements equal to $h_{2}$, and then swapping the four elements
%that we changed with four unchanged elements.}}{\Qlb{swapfig}}{swapfig.eps}%
%width 3.7343in;  height 1.1398in;  depth 0pt;  original-width 10.6017in;
%original-height 8.0782in;  cropleft "0.1939";  croptop "0.9591";
%cropright "0.7671";  cropbottom "0.7334";
%filename 'swapfig.eps';file-properties "XNPEU";}} }%
%BeginExpansion
\begin{figure}[ptb]%
\centering
\includegraphics[
trim=2.055670in 5.924552in 2.469136in 0.330398in,
height=1.1398in,
width=3.7343in
]%
{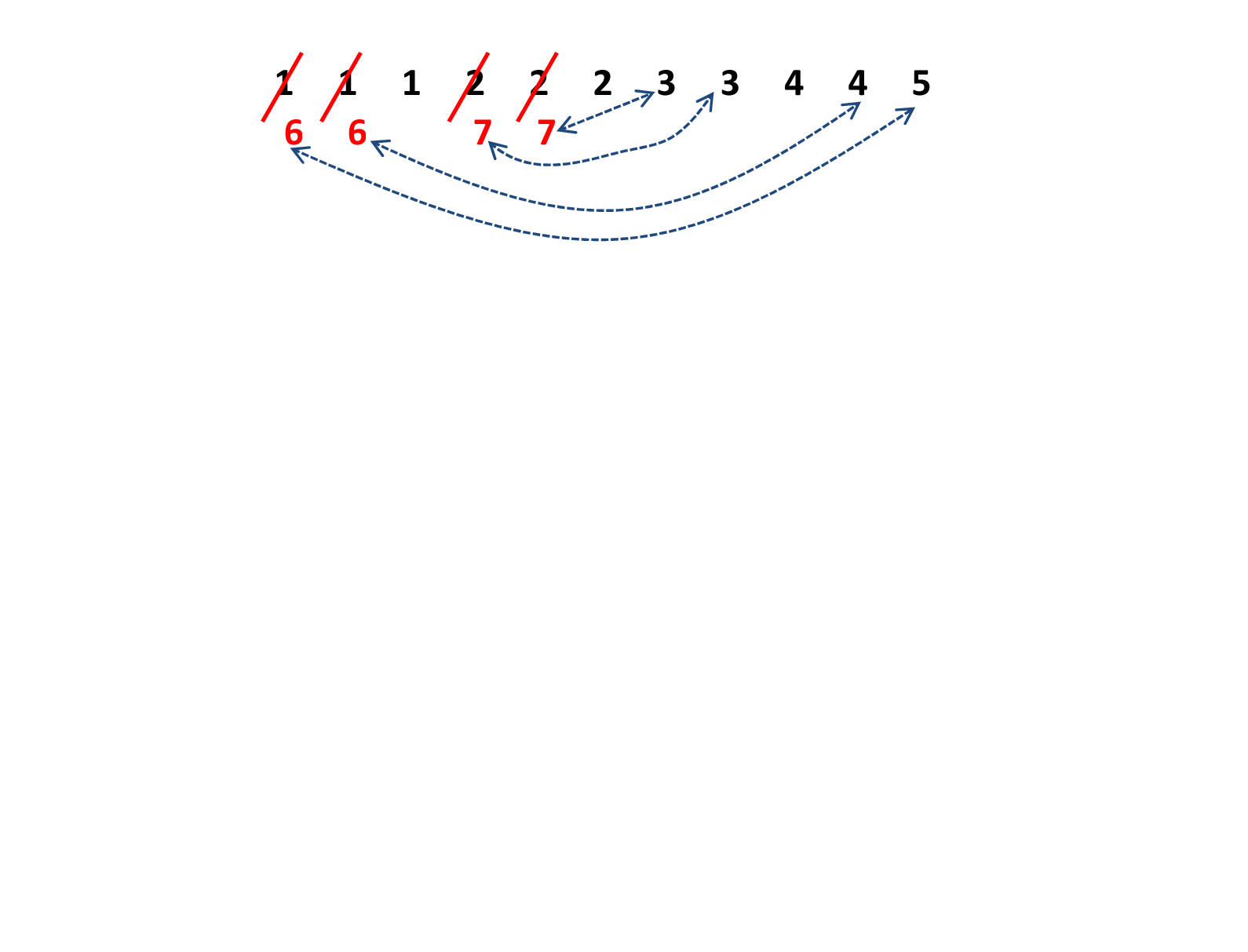}%
\caption{In this example, $N=11$, $r=2$, $P/2^{\ell}=2$, and $a_{1}=a_{2}=3$.
\ So we transform $X$ to $Y$ by choosing $h_{1}=1$\ and $h_{2}=2$, changing
any two elements equal to $h_{1}$ and any two elements equal to $h_{2}$, and
then swapping the four elements that we changed with four unchanged elements.}%
\label{swapfig}%
\end{figure}
%EndExpansion
Note that we can reverse the procedure in a natural way to go from $Y$ back to
$X$:

\begin{enumerate}
\item[(1)] For each $j\in\left[  r\right]  $, change exactly $P/2^{\ell
^{\prime}}$\ of the $x_{i}$'s that are equal to $h_{j}^{\prime}$ to value
$h_{j}$.

\item[(2)] Swap the $d$\ elements of $X$\ that were changed in step (2) with
any $d$\ elements $x_{i}$ of $X$, subject to the same constraints as in the
step (2) of the $X\rightarrow Y$ conversion.
\end{enumerate}

Fix any $\left(  X,Y\right)  \in R$, and let $i\in\left[  N\right]  $\ be any
index such that $x_{i}\neq y_{i}$. \ Then we claim that the parameters of
Theorem \ref{ambthm} satisfy either $q_{X,i}\leq\frac{6d}{N-d}$\ or
$q_{Y,i}\leq\frac{6d}{N-d}$. \ To see this, let us write $q_{X,i}%
=q_{X,i}^{\prime}+q_{X,i}^{\prime\prime}$, where $q_{X,i}^{\prime}$ is the
probability that $x_{i}$ is changed in step (1) of the $X\rightarrow
Y$\ conversion and $q_{X,i}^{\prime\prime}$ is the probability that $x_{i}$ is
not changed in step (1), but is swapped with some changed element in step (2).
\ We also express $q_{Y,i}$ in a similar way, with respect to the
$Y\rightarrow X$\ conversion.

We consider two cases. \ The first case is that $x_{i}$\ is one of the
\textquotedblleft other $d$\ elements\textquotedblright\ with which we swap
the changed elements in step (2) of the $X\rightarrow Y$\ conversion. \ In
this case, $q_{X,i}^{\prime}\neq0$ only if $x_{i}=h_{j}$ for some $j$. \ Then
because of the constraint (a), we have $q_{X,i}^{\prime}\leq\frac{3d}{N+2d}$.
\ We also have
\[
q_{X,i}^{\prime\prime}=\Pr_{Y^{\prime}\in\mathcal{A}_{\ell}}\left[  x_{i}\neq
y_{i}^{\prime}~|~\left(  X,Y^{\prime}\right)  \in R\right]  \leq\frac
{d}{(N-d)/3}=\frac{3d}{N-d},
\]
because each of the constraints (a) and (b) eliminates at most $(N-d)/3$ of
the $N-d$ variables $x_{i}$ that are available for swapping in step (2).
\ Therefore, $q_{X,i}=q_{X,i}^{\prime}+q_{X,i}^{\prime\prime}\leq\frac
{6d}{N-d}$.

The second case is that $x_{i}$ is one of the elements that are changed in
step (1) of the $X\rightarrow Y$\ conversion. \ Then $y_{i}$\ is one of the
\textquotedblleft other $d$\ elements\textquotedblright\ in step (2) of the
$Y\rightarrow X$\ conversion. \ Similarly to the previous case, we can show
that $q_{Y,i}\leq\frac{6d}{N-d}$.

Since $q_{X,i}\leq1$\ and $q_{Y,i}\leq1$, it follows that%
\[
q_{X,i}q_{Y,i}\leq\frac{6d}{N-d}.
\]
Thus, by Theorem \ref{ambthm},%
\[
\operatorname*{Q}(\mathcal{A}_{\ell},\mathcal{A}_{\ell-1})=\Omega\left(
\frac{1}{\sqrt{q_{X,i}q_{Y,i}}}\right)  =\Omega\left(  \sqrt{\frac{N-d}{d}%
}\right)  =\Omega\left(  \sqrt{\frac{N}{d}}\right)  .
\]

\end{proof}

We now consider the case that many elements are chopped. \ Here we prove a
lower bound by reduction from \textsc{SetEquality}. \ Given two sequences of
integers $Y\in\left[  M\right]  ^{N}$\ and $Z\in\left[  M\right]  ^{N}$,
neither with any repeats, the \textsc{SetEquality}\ problem is to decide
whether $Y$ and $Z$ are equal as sets or disjoint as sets, promised that one
of these is the case. \ \textsc{SetEquality}\ is similar to the collision
problem studied by Aaronson and Shi\ \cite{as}, but it lacks permutation
symmetry, making it harder to prove a lower bound by the polynomial method.
\ By combining the collision lower bound with Ambainis's adversary method,
Midrijanis \cite{midrijanis}\ was nevertheless able to show that%
\[
\operatorname*{Q}\left(  \text{\textsc{SetEquality}}\right)  =\Omega\left(
\left(  \frac{N}{\log N}\right)  ^{1/5}\right)  .
\]
Very recently, and using different ideas, Zhandry \cite{zhandry:col}\ managed
to improve Midrijanis's lower bound to the following:

\begin{theorem}
[Zhandry \cite{zhandry:col}]\label{zhandrythm}$\operatorname*{Q}\left(
\text{\textsc{SetEquality}}\right)  =\Omega(N^{1/3})$.
\end{theorem}

Theorem \ref{zhandrythm}\ is known to be tight, by the upper bound of
Brassard, H\o yer, and Tapp \cite{bht} mentioned in Section \ref{FORMAL}.

We will consider a modification of the \textsc{SetEquality}\ problem, which we
call \textsc{3SetEquality}. \ Here we are given three sequences of integers
$Y,Z,W\in\left[  M\right]  ^{N}$, none of which has any repeats. \ We are
promised that $Y$ and $W$ are disjoint as sets, and that $Z$ is equal either
to $Y$ or to $W$ as a set. \ The task is to distinguish between those two cases.

\begin{theorem}
\label{zhandrythm1}$\operatorname*{Q}($\textsc{3SetEquality}$)=\Omega
(N^{1/3})$.
\end{theorem}

\begin{proof}
The theorem follows from Theorem \ref{zhandrythm} together with the following
claim: if \textsc{3SetEquality} is solvable by a quantum algorithm
$\mathcal{A}$ that uses $T$ queries, then \textsc{3SetEquality} is solvable by
a quantum algorithm that uses $O(T)$ queries.

To show this, let $Y,W$ be an instance of \textsc{SetEquality}. \ We produce
an instance of \textsc{3SetEquality} by choosing $Z$ to be either a randomly
permuted version of $Y$ or a randomly permuted version of $W$. \ We then run
the algorithm for \textsc{3SetEquality} on that instance. \ If $Y$ and $W$ are
disjoint, then the promise of \textsc{3SetEquality} is satisfied and the
algorithm will find whether we used $Y$ or $W$ to generate $Z$. \ If $Y=W$,
then using $Y$ and using $W$ results in the same probability distribution for
$Z$; hence no algorithm will be able to guess whether we used $Y$ or $W$ with
probability greater than $1/2$.
\end{proof}

We now use Theorem \ref{zhandrythm1} to prove another lower bound on
$\operatorname*{Q}(\mathcal{A}_{\ell},\mathcal{A}_{\ell-1})$.

\begin{lemma}
\label{usingzha}Suppose $\mathcal{A}_{\ell}$\ was formed from $\mathcal{A}%
_{\ell-1}$\ by chopping $r$ rows.\ \ Then $\operatorname*{Q}\left(
\mathcal{A}_{\ell},\mathcal{A}_{\ell-1}\right)  =\Omega\left(  r^{1/3}\right)
$.
\end{lemma}

\begin{proof}
We will show how to embed a \textsc{3SetEquality}\ instance of size $r$ into
the $\mathcal{A}_{\ell}$\ versus $\mathcal{A}_{\ell-1}$\ problem.

Let $\mathcal{A}_{\ell-1}=(a_{1},\ldots,a_{u})$. \ Also, let $i\left(
1\right)  ,\ldots,i\left(  r\right)  \in\left[  u\right]  $\ be the $r$ rows
that are chopped in going from $\mathcal{A}_{\ell-1}$\ to $\mathcal{A}_{\ell}%
$, let $i^{\prime}\left(  1\right)  ,\ldots,i^{\prime}\left(  r\right)
\in\left[  u\right]  $\ be the $r$ rows that are augmented, and let $j\left(
1\right)  ,\ldots,j\left(  u-2r\right)  \in\left[  u\right]  $\ be the
$u-2r$\ rows that are left unchanged. \ Recall that, in going from
$\mathcal{A}_{\ell-1}$\ to $\mathcal{A}_{\ell}$, each row $i\left(  k\right)
$ (or $i^{\prime}\left(  k\right)  $)\ is chopped or augmented by
$P/2^{\ell^{\prime}}$ elements, where $\ell^{\prime}=\lceil\frac{\ell}%
{2}\rceil$.

Now let $Y=(y_{1},\ldots,y_{r})$, $Z=(z_{1},\ldots,z_{r})$, $W=(z_{1}%
,\ldots,z_{r})$ be an instance of \textsc{3SetEquality}. \ Then we construct
an input $X\in\left[  M\right]  ^{N}$\ as follows. \ First, for each
$k\in\left[  r\right]  $,\ set $a_{i\left(  k\right)  }-P/2^{\ell^{\prime}}%
$\ of the $x_{i}$'s equal to $y_{k}$, set $P/2^{\ell^{\prime}}$\ of the
$x_{i}$'s equal to\ $z_{k}$ and set $a_{i\left(  k\right)  }^{\prime}$ of the
$x_{i}$'s equal to $w_{k}$. \ Next, let $w_{1},w_{2},\ldots\in\left[
M\right]  $ be a list of numbers that are guaranteed \textit{not} to be in
$Y\cup Z$. \ Then for each $k\in\left[  u-2r\right]  $, set $a_{j\left(
k\right)  }$\ of the $x_{i}$'s\ equal to $w_{k}$.

It is easy to see that, if $Y$ and $Z$\ are equal as sets, then $X$\ will have
type $\mathcal{A}_{\ell-1}$, while if $Z$\ and $W$\ are\ equal as sets, then
$X$\ will have type $\mathcal{A}_{\ell}$. \ So in deciding whether
$X$\ belongs to $\mathcal{A}_{\ell}$\ or $\mathcal{A}_{\ell-1}$, we also
decide whether $Y=Z$ or $Z=W$. \ The lemma now follows from Theorem
\ref{zhandrythm1}.
\end{proof}

\subsection{Putting Everything Together\label{TOGETHER}}

Let $\mathcal{C}$\ be a\ quantum query algorithm that distinguishes
$\mathcal{A}_{0}=\mathcal{A}^{\ast}$\ from $\mathcal{A}_{2L}=\mathcal{B}%
^{\ast}$, and assume $\mathcal{C}$\ is optimal: that is, it makes
$\operatorname*{Q}(\mathcal{A}^{\ast},\mathcal{B}^{\ast})\leq\operatorname*{Q}%
(f)$\ queries. \ As mentioned earlier, we can assume that\ $\Pr\left[
\mathcal{C}~\text{accepts }X\right]  $\ depends only on the type of $X$.
\ Thus, let%
\begin{align*}
p_{\ell}  &  :=\Pr\left[  \mathcal{C}~\text{accepts }X\in\mathcal{A}_{\ell
}\right]  .
\end{align*}
Then by assumption, $\left\vert p_{0}-p_{2L}\right\vert \geq1/3$. Now let
$\beta_{\ell}:=\frac{1}{10\ell^{2}}$, and observe that $\sum_{\ell=1}^{\infty
}\beta_{\ell}<\frac{1}{6}$. \ By the triangle inequality, it follows that
there exists an $\ell\in\left[  2L\right]  $\ such that $\left\vert p_{\ell
}-p_{\ell-1}\right\vert \geq\beta_{\ell}$. \ In other words, we get a
$\operatorname*{Q}(f)$-query\ quantum algorithm that distinguishes
$\mathcal{A}_{\ell}$ from $\mathcal{A}_{\ell-1}$\ with bias $\beta_{\ell}$.
\ By Proposition \ref{amplify}, this immediately implies%
\[
\operatorname*{Q}(\mathcal{A}_{\ell},\mathcal{A}_{\ell-1})=O\left(
\frac{\operatorname*{Q}(f)}{\beta_{\ell}}\right)
\]
or equivalently%
\[
\operatorname*{Q}(f)=\Omega\left(  \frac{\operatorname*{Q}(\mathcal{A}_{\ell
},\mathcal{A}_{\ell-1})}{\ell^{2}}\right)  .
\]
Now let $d=\left\Vert \mathcal{A}_{\ell}-\mathcal{A}_{\ell-1}\right\Vert $,
and suppose $\mathcal{A}_{\ell}$\ was produced from $\mathcal{A}_{\ell-1}$\ by
chopping $r$ rows. \ Then $d=rP/2^{\ell^{\prime}}\leq2rN/2^{\ell^{\prime}}$
where $l^{\prime}=\lceil\frac{l}{2}\rceil$. \ Combining Lemmas \ref{usingamb}
and \ref{usingzha}, we find that%
\begin{align*}
\operatorname*{Q}(\mathcal{A}_{\ell},\mathcal{A}_{\ell-1})  &  =\Omega\left(
\max\left\{  \sqrt{\frac{N}{d}},r^{1/3}\right\}  \right) \\
&  =\Omega\left(  \sqrt{\frac{2^{\ell^{\prime}}}{r}}+r^{1/3}\right) \\
&  =\Omega\left(  2^{\ell^{\prime}/5}\right)  ,
\end{align*}
since the minimum occurs when $r$ is asymptotically $2^{3\ell^{\prime}/5}$.
\ If $\ell^{\prime}\leq\left(  \log_{2}T\right)  -2$, then combining Lemmas
\ref{usingamb} and \ref{chopub}, we also have the lower bound%
\[
\operatorname*{Q}(\mathcal{A}_{\ell},\mathcal{A}_{\ell-1})=\Omega\left(
\sqrt{\frac{N}{4N/T^{c}}}\right)  =\Omega(\sqrt{T^{c}}).
\]
Hence%
\[
\operatorname*{Q}(f)=\left\{
\begin{array}
[c]{cc}%
\Omega\left(  \frac{\sqrt{T^{c}}}{\ell^{2}}\right)  & \text{if }\ell^{\prime
}\leq\left(  \log_{2}T\right)  -2\\
\Omega\left(  2^{\ell^{\prime}/5}\right)  & \text{if }\ell^{\prime}>\left(
\log_{2}T\right)  -2.
\end{array}
\right.
\]
Let us now make the choice $c=2/5$,\ so that we get a lower bound of%
\[
\operatorname*{Q}(f)=\Omega\left(  \frac{T^{1/5}}{\log^{2}T}\right)
\]
in either case. \ Hence $T=O(\operatorname*{Q}(f)^{5}\log^{10}%
\operatorname*{Q}(f))$. \ By Lemma \ref{rworks}:%
\begin{align*}
\operatorname*{R}(f)  &  =O(T^{1+c}\log T)\\
&  =O(T^{7/5}\log T)\\
&  =O(\operatorname*{Q}(f)^{7}\log^{15}\operatorname*{Q}(f)).
\end{align*}
This completes the proof of Theorem \ref{mainthm}.

\section{Quantum Lower Bounds Under The Uniform Distribution\label{RO}}

In this section, we consider the problems of $\mathsf{P}\overset{?}{=}%
\mathsf{BQP}$\ relative to a random oracle,\ and of simulating a $T$-query
quantum algorithm on \textit{most} inputs using $T^{O\left(  1\right)  }%
$\ classical queries. \ We show that these problems are connected to a
fundamental conjecture about influences in low-degree polynomials.

Recall Conjecture \ref{infconj},\ which said that \textit{bounded polynomials
have influential variables}: that is, for every degree-$d$ polynomial
$p:\mathbb{R}^{N}\rightarrow\mathbb{R}$\ such that $0\leq p(X)\leq1$\ for all
$X\in\left\{  0,1\right\}  ^{N}$, there exists an $i\in\left[  N\right]  $
such that $\operatorname*{Inf}\nolimits_{i}\left[  p\right]  \geq
(\operatorname*{Var}\left[  p\right]  /d)^{O\left(  1\right)  }$, where%
\begin{align*}
\operatorname*{Inf}\nolimits_{i}\left[  p\right]   &  :=\operatorname*{E}%
_{X\in\left\{  0,1\right\}  ^{N}}\left[  (p(X)-p(X^{i}))^{2}\right]  ,\\
\operatorname*{Var}\left[  p\right]   &  :=\operatorname*{E}_{X\in\left\{
0,1\right\}  ^{N}}\left[  (p(X)-\operatorname*{E}\left[  p\right]
)^{2}\right]  .
\end{align*}
We will show that Conjecture \ref{infconj} has several powerful consequences
for quantum complexity theory.

As a first step, let%
\[
\operatorname*{Inf}\left[  p\right]  :=\sum_{i=1}^{N}\operatorname*{Inf}%
\nolimits_{i}\left[  p\right]
\]
be the \textit{total influence} of $p$. \ Then we have the following bound,
versions of which have long been known in the analysis of Boolean functions
community,\footnote{For example, Shi \cite{shi:inf} proved the bound for the
special case of Boolean functions, and generalizing his proof to arbitrary
bounded functions is straightforward.} but which we prove for completeness.

\begin{lemma}
[folklore]\label{qlbinf}Let $p:\mathbb{R}^{N}\rightarrow\mathbb{R}$\ be a
degree-$d$ real polynomial such that $0\leq p(X)\leq1$\ for all $X\in\left\{
0,1\right\}  ^{N}$. \ Then $\operatorname*{Inf}\left[  p\right]  \leq d$.
\end{lemma}

\begin{proof}
Let $q$ be the analogue of $p$ in the Fourier representation:%
\[
q(x_{1},\ldots,x_{N}):=1-2p\left(  \frac{1-x_{1}}{2},\ldots,\frac{1-x_{N}}%
{2}\right)  .
\]
Clearly $\deg(q)=\deg(p)=d$\ and $-1\leq q(X)\leq1$\ for all $X\in\left\{
1,-1\right\}  ^{N}$. \ Also, defining $X^{i}$\ to be $X\in\left\{
1,-1\right\}  ^{N}$\ with $x_{i}$ negated, and%
\[
\operatorname*{Inf}\nolimits_{i}\left[  q\right]  :=\frac{1}{4}%
\operatorname*{E}_{X\in\left\{  1,-1\right\}  ^{N}}\left[  (q(X)-q(X^{i}%
))^{2}\right]  ,
\]
we have $\operatorname*{Inf}\nolimits_{i}\left[  q\right]
=\operatorname*{Inf}\nolimits_{i}\left[  p\right]  $.

Note that we can express $q$ as%
\[
q(X)=\sum_{S\subseteq\left[  N\right]  ~:~\left\vert S\right\vert \leq
d}\alpha_{S}\chi_{S}(X),
\]
where $\alpha_{S}\in\mathbb{R}$\ and $\chi_{S}(X):=\prod_{i\in S}x_{i}$\ is
the Fourier character corresponding to the set $S$. \ Furthermore, by
Parseval's identity,%
\[
\sum_{\left\vert S\right\vert \leq d}\alpha_{S}^{2}=\frac{1}{2^{N}}\sum
_{X\in\left\{  1,-1\right\}  ^{N}}q(X)^{2}\leq1.
\]
Now, in the Fourier representation, it is known that%
\[
\operatorname*{Inf}\nolimits_{i}\left[  q\right]  =\sum_{\left\vert
S\right\vert \leq d~:~i\in S}\alpha_{S}^{2}.
\]
Hence%
\[
\operatorname*{Inf}\left[  p\right]  =\operatorname*{Inf}\left[  q\right]
=\sum_{i\in\left[  N\right]  }\sum_{\left\vert S\right\vert \leq d~:~i\in
S}\alpha_{S}^{2}=\sum_{\left\vert S\right\vert \leq d~}\sum_{i\in S}\alpha
_{S}^{2}=\sum_{\left\vert S\right\vert \leq d}\left\vert S\right\vert
\alpha_{S}^{2}\leq d\sum_{\left\vert S\right\vert \leq d}\alpha_{S}^{2}\leq d
\]
as claimed.
\end{proof}

We also need the following lemma of Beals et al.\ \cite{bbcmw}.

\begin{lemma}
[Beals et al.]\label{bbcmwlem}Suppose a quantum algorithm $Q$\ makes $T$
queries to a Boolean input $X\in\left\{  0,1\right\}  ^{N}$. \ Then $Q$'s
acceptance probability is a real multilinear polynomial $p(X)$, of degree at
most $2T$.
\end{lemma}

\subsection{Consequences of Our Influence Conjecture\label{CONSEQ}}

We now prove our first consequence of Conjecture \ref{infconj}: namely, that
it implies the folklore Conjecture \ref{folkloreconj}.

\begin{theorem}
\label{impthm}Suppose Conjecture \ref{infconj} holds, and let $\varepsilon
,\delta>0$. \ Then given any quantum algorithm $Q$ that makes $T$ queries to a
Boolean input $X$, there exists a deterministic classical algorithm that makes
$\operatorname*{poly}(T,1/\varepsilon,1/\delta)$\ queries,\ and that
approximates $Q$'s acceptance probability to within an additive constant
$\varepsilon$ on a $1-\delta$\ fraction of inputs.
\end{theorem}

\begin{proof}
Let $p(X)$\ be the probability that $Q$ accepts input $X=(x_{1},\ldots,x_{N}%
)$.\ \ Then Lemma \ref{bbcmwlem} says that $p$ is a real polynomial of degree
at most $2T$. \ Assume Conjecture \ref{infconj}. \ Then for every such $p$,
there exists an index $i$ satisfying $\operatorname*{Inf}\nolimits_{i}\left[
p\right]  \geq w(\operatorname*{Var}\left[  p\right]  /T)$, for some fixed
polynomial $w$. \ Under that assumption, we give a classical algorithm
$C$\ that makes $\operatorname*{poly}(T,1/\varepsilon,1/\delta)$\ queries to
the $x_{i}$'s, and that approximates $p(X)$\ on most inputs $X$. \ In what
follows, assume $X\in\left\{  0,1\right\}  ^{N}$ is uniformly random.\bigskip

\texttt{\qquad set }$p_{0}:=p$

\texttt{\qquad for }$j:=0,1,2,\ldots$\texttt{:}

\texttt{\qquad\qquad if }$\operatorname*{Var}\left[  p_{j}\right]
\leq\varepsilon^{2}\delta/2$

\texttt{\qquad\qquad\qquad output }$\operatorname*{E}_{Y\in\left\{
0,1\right\}  ^{N-j}}\left[  p_{j}(Y)\right]  $\texttt{\ as approximation for
}$p(X)$\texttt{\ and halt}

\texttt{\qquad\qquad else}

\texttt{\qquad\qquad\qquad find an }$i\in\left[  N-j\right]  $\texttt{\ such
that }$\operatorname*{Inf}\nolimits_{i}\left[  p_{j}\right]  >w(\varepsilon
^{2}\delta/2T)$

\texttt{\qquad\qquad\qquad query }$x_{i}$\texttt{, and let }$p_{j+1}%
:\mathbb{R}^{N-j}\rightarrow\mathbb{R}$\texttt{\ be the polynomial}

\texttt{\qquad\qquad\qquad\qquad induced by the answer\bigskip}

When $C$\ halts, by assumption $\operatorname*{Var}\left[  p_{j}\right]
\leq\varepsilon^{2}\delta/2$. \ By Markov's inequality, this implies%
\[
\Pr_{X\in\left\{  0,1\right\}  ^{N-j}}\left[  \left\vert p_{j}%
(X)-\operatorname*{E}\left[  p_{j}\right]  \right\vert >\varepsilon\right]
<\frac{\delta}{2},
\]
meaning that \textit{when} $C$ halts, it succeeds with probability at least
$1-\delta/2$.

On the other hand, suppose $\operatorname*{Var}\left[  p_{j}\right]
>\varepsilon^{2}\delta/2$. \ Then by Conjecture \ref{infconj}, there exists an
index $i^{\ast}\in\left[  N\right]  $\ such that%
\[
\operatorname*{Inf}\nolimits_{i^{\ast}}\left[  p_{j}\right]  \geq w\left(
\frac{\operatorname*{Var}\left[  p_{j}\right]  }{T}\right)  \geq w\left(
\frac{\varepsilon^{2}\delta}{2T}\right)  .
\]
Thus, suppose we query $x_{i^{\ast}}$. \ Since $X$ is uniformly random,
$x_{i^{\ast}}$\ will be $0$ or $1$ with equal probability, even conditioned on
the results of all previous queries. \ So after the query, our new polynomial
$p_{j+1}$\ will satisfy%
\[
\Pr\left[  p_{j+1}=p_{j|x_{i^{\ast}}=0}\right]  =\Pr\left[  p_{j+1}%
=p_{j|x_{i^{\ast}}=1}\right]  =\frac{1}{2},
\]
where $p_{j|x_{i^{\ast}}=0}$\ and $p_{j|x_{i^{\ast}}=1}$\ are the polynomials
on $N-j-1$\ variables obtained from $p_{j}$\ by restricting $x_{i^{\ast}}$\ to
$0$ or $1$ respectively. \ Therefore%
\begin{align*}
\operatorname*{E}_{x_{i^{\ast}}\in\left\{  0,1\right\}  }\left[
\operatorname*{Inf}\left[  p_{j+1}\right]  \right]   &  =\frac{1}{2}\left(
\operatorname*{Inf}\left[  p_{j|x_{i^{\ast}}=0}\right]  +\operatorname*{Inf}%
\left[  p_{j|x_{i^{\ast}}=1}\right]  \right) \\
&  =\frac{1}{2}\left(  \sum_{i\neq i^{\ast}}\operatorname*{Inf}\nolimits_{i}%
\left[  p_{j|x_{i^{\ast}}=0}\right]  +\sum_{i\neq i^{\ast}}\operatorname*{Inf}%
\nolimits_{i}\left[  p_{j|x_{i^{\ast}}=1}\right]  \right) \\
&  =\sum_{i\neq i^{\ast}}\operatorname*{Inf}\nolimits_{i}\left[  p_{j}\right]
\\
&  =\operatorname*{Inf}\left[  p_{j}\right]  -\operatorname*{Inf}%
\nolimits_{i^{\ast}}\left[  p_{j}\right] \\
&  \leq\operatorname*{Inf}\left[  p_{j}\right]  -w\left(  \frac{\varepsilon
^{2}\delta}{2T}\right)  .
\end{align*}
By linearity of expectation, this imples that for all $j$,%
\[
\operatorname*{E}_{X\in\left\{  0,1\right\}  ^{N}}\left[  \operatorname*{Inf}%
\left[  p_{j}\right]  \right]  \leq\operatorname*{Inf}\left[  p_{0}\right]
-jw\left(  \frac{\varepsilon^{2}\delta}{2T}\right)
\]
But recall from Lemma \ref{qlbinf} that%
\[
\operatorname*{Inf}\left[  p_{0}\right]  \leq\deg\left(  p_{0}\right)
\leq2T.
\]
It follows that $C$ halts after an expected number of iterations that is at
most%
\[
\frac{\operatorname*{Inf}\left[  p_{0}\right]  }{w(\varepsilon^{2}\delta
/2T)}\leq\frac{2T}{w(\varepsilon^{2}\delta/2T)}.
\]
Thus, by Markov's inequality, the probability (over $X$) that $C$ has
\textit{not} halted after $\frac{4T}{\delta\cdot w(\varepsilon^{2}\delta/2T)}%
$\ iterations is at most $\delta/2$. \ Hence by the union bound, the
probability over $X$ that $C$ fails is at most $\delta/2+\delta/2=\delta$.
\ Since each iteration queries exactly one variable and%
\[
\frac{4T}{\delta\cdot w(\varepsilon^{2}\delta/2T)}=\operatorname*{poly}%
(T,1/\varepsilon,1/\delta),
\]
this completes the proof.
\end{proof}

An immediate corollary is the following:

\begin{corollary}
\label{dqcor}Suppose Conjecture \ref{infconj} holds.
\ Then$\ \operatorname*{D}_{\varepsilon+\delta}(f)\leq(\operatorname*{Q}%
_{\varepsilon}(f)/\delta)^{O\left(  1\right)  }$ for all Boolean functions
$f$\ and all $\varepsilon,\delta>0$.
\end{corollary}

\begin{proof}
Let $Q$ be a quantum algorithm that evaluates $f(X)$, with bounded error, on a
$1-\varepsilon$\ fraction of inputs $X\in\left\{  0,1\right\}  ^{N}$. \ Let
$p(X):=\Pr\left[  Q\text{ accepts }X\right]  $. \ Now run the classical
simulation algorithm $C$ from Theorem \ref{impthm}, to obtain an estimate
$\widetilde{p}(X)$\ of $p(X)$\ such that%
\[
\Pr_{X\in\left\{  0,1\right\}  ^{N}}\left[  \left\vert \widetilde{p}%
(X)-p(X)\right\vert \leq\frac{1}{10}\right]  \geq1-\delta.
\]
Output $f(X)=1$\ if\ $\widetilde{p}(X)\geq\frac{1}{2}$\ and $f(X)=0$%
\ otherwise. \ By the theorem, this requires $\operatorname*{poly}%
(T,1/\delta)$ queries to $X$, and by the union bound it successfully computes
$f(X)$\ on at least a $1-\varepsilon-\delta$\ fraction of inputs $X$.
\end{proof}

We also get the following complexity-theoretic consequence:

\begin{theorem}
\label{pbqpthm}Suppose Conjecture \ref{infconj} holds. \ Then\ $\mathsf{P}%
=\mathsf{P}^{\mathsf{\#P}}$\ implies $\mathsf{BQP}^{A}\subset\mathsf{AvgP}%
^{A}$\ \ with probability $1$ for a random oracle $A$.
\end{theorem}

\begin{proof}
Let $Q$ be a polynomial-time quantum Turing machine that queries an oracle
$A$, and assume $Q$\ decides some language $L\in\mathsf{BQP}^{A}$\ with
bounded error. \ Given an input $x\in\left\{  0,1\right\}  ^{n}$, let
$p_{x}(A):=\Pr\left[  Q^{A}(x)\text{ accepts}\right]  $. \ Then clearly
$p_{x}(A)$\ depends only on some finite prefix $B$\ of $A$, of size
$N=2^{\operatorname*{poly}(n)}$. \ Furthermore, Lemma \ref{bbcmwlem} implies
that $p_{x}$\ is a polynomial in the bits of $B$, of degree at most
$\operatorname*{poly}(n)$.

Assume Conjecture \ref{infconj} as well as $\mathsf{P}=\mathsf{P}%
^{\mathsf{\#P}}$. \ Then we claim that there exists a deterministic
polynomial-time algorithm $C$ such that for all $Q$ and $x\in\left\{
0,1\right\}  ^{n}$,%
\begin{equation}
\Pr_{A}\left[  \left\vert \widetilde{p}_{x}(A)-p_{x}(A)\right\vert >\frac
{1}{10}\right]  <\frac{1}{n^{3}}, \label{apx}%
\end{equation}
where $\widetilde{p}_{x}(A)$\ is the output of $C$ given input $x$ and oracle
$A$. \ This $C$ is essentially just the algorithm from Theorem \ref{impthm}.
\ The key point is that we can implement $C$ using not only
$\operatorname*{poly}(n)$\ queries to $A$, but also $\operatorname*{poly}%
(n)$\ computation steps.

To prove the claim, let $M$\ be any of the $2^{\operatorname*{poly}(n)}$
monomials\ in the polynomial $p_{j}$ from Theorem \ref{impthm}, and let
$\alpha_{M}$\ be the coefficient of $M$. \ Then notice that $\alpha_{M}$\ can
be computed to $\operatorname*{poly}(n)$\ bits of precision in $\mathsf{P}%
^{\mathsf{\#P}}$, by the same techniques used to show $\mathsf{BQP}%
\subseteq\mathsf{P}^{\mathsf{\#P}}$ \cite{bv}. \ Therefore the expectation%
\[
\operatorname*{E}_{Y\in\left\{  0,1\right\}  ^{N-j}}\left[  p_{j}(Y)\right]
=\sum_{M}\frac{\alpha_{M}}{2^{\left\vert M\right\vert }}%
\]
can be computed in $\mathsf{P}^{\mathsf{\#P}}$\ as well. \ The other two
quantities that arise in the algorithm---$\operatorname*{Var}\left[
p_{j}\right]  $\ and $\operatorname*{Inf}\nolimits_{i}\left[  p_{j}\right]
$---can also be computed in $\mathsf{P}^{\mathsf{\#P}}$, since they are simply
sums of squares of differences of $p_{j}(X)$'s. \ This means that finding an
$i$ such that $\operatorname*{Inf}\nolimits_{i}\left[  p_{j}\right]
>w(\varepsilon^{2}\delta/T)$\ is in $\mathsf{NP}^{\mathsf{\#P}}$. \ But under
the assumption that $\mathsf{P}=\mathsf{P}^{\mathsf{\#P}}$, we have
$\mathsf{P}=\mathsf{NP}^{\mathsf{\#P}}$ as well. \ Therefore all of the
computations needed to implement $C$\ take polynomial time.

Now\ let $\delta_{n}(A)$\ be the fraction of inputs $x\in\left\{  0,1\right\}
^{n}$\ such that $\left\vert \widetilde{p}_{x}(A)-p_{x}(A)\right\vert
>\frac{1}{10}$. \ Then by (\ref{apx}) together with Markov's inequality,%
\[
\Pr_{A}\left[  \delta_{n}(A)>\frac{1}{n}\right]  <\frac{1}{n^{2}}.
\]
Since $\sum_{n=1}^{\infty}\frac{1}{n^{2}}$\ converges, it follows that
$\delta_{n}(A)\leq\frac{1}{n}$\ for all but finitely many values of $n$, with
probability $1$ over $A$. \ Assuming this occurs, we can simply hardwire the
behavior of $Q$\ on the remaining $n$'s into our classical simulation
procedure $C$. \ Hence $L\in\mathsf{AvgP}^{A}$.

Since the number of $\mathsf{BQP}^{A}$ languages is countable, the above
implies that $L\in\mathsf{AvgP}^{A}$ for every $L\in\mathsf{BQP}^{A}%
$\ \textit{simultaneously} (that is, $\mathsf{BQP}^{A}\subset\mathsf{AvgP}%
^{A}$) with probability $1$ over $A$.
\end{proof}

As a side note, suppose we had an extremely strong variant of Conjecture
\ref{infconj}, one that implied something like%
\[
\Pr_{A}\left[  \left\vert \widetilde{p}_{x}(A)-p_{x}(A)\right\vert >\frac
{1}{10}\right]  <\frac{1}{\exp(n)}.
\]
in place of (\ref{apx}). \ Then we could eliminate the need for $\mathsf{AvgP}%
$ in Theorem \ref{pbqpthm}, and show that $\mathsf{P}=\mathsf{P}%
^{\mathsf{\#P}}$\ implies $\mathsf{P}^{A}=\mathsf{BQP}^{A}$\ \ with
probability $1$ for a random oracle $A$.

\subsection{Unconditional Results\label{UNCOND}}

We conclude this section with some unconditional results. \ These results will
use Theorem \ref{dfkothm0} of Dinur et al.\ \cite{dfko}: that for every
degree-$d$ polynomial $p:\mathbb{R}^{N}\rightarrow\mathbb{R}$\ such that
$0\leq p(X)\leq1$\ for all $X\in\left\{  0,1\right\}  ^{N}$, there exists a
polynomial $\widetilde{p}$ depending on at most $2^{O\left(  d\right)
}/\varepsilon^{2}$\ variables\ such that $\left\Vert \widetilde{p}%
-p\right\Vert _{2}^{2}\leq\varepsilon$, where%
\[
\left\Vert p\right\Vert _{2}^{2}:=\operatorname*{E}_{X\in\left\{  0,1\right\}
^{N}}\left[  p(X)^{2}\right]  .
\]
Theorem \ref{dfkothm0}\ has the following simple corollary.

\begin{corollary}
\label{nonadap}Suppose a quantum algorithm $Q$ makes $T$ queries to a Boolean
input $X\in\left\{  0,1\right\}  ^{N}$. \ Then for all $\alpha,\delta>0$, we
can approximate $Q$'s acceptance probability to within an additive constant
$\alpha$, on a $1-\delta$\ fraction of inputs, by making $\frac{2^{O\left(
T\right)  }}{\alpha^{4}\delta^{4}}$\ deterministic classical queries to $X$.
\ (Indeed, the classical queries are nonadaptive.)
\end{corollary}

\begin{proof}
Let $p(X):=\Pr\left[  Q\text{ accepts }X\right]  $. \ Then $p$\ is a
degree-$2T$ real polynomial by Lemma \ref{bbcmwlem}. \ Hence, by Theorem
\ref{dfkothm0}, there exists a polynomial $\widetilde{p}$, depending on
$K=\frac{2^{O\left(  T\right)  }}{\alpha^{4}\delta^{4}}$\ variables $x_{i_{1}%
},\ldots,x_{i_{K}}$, such that%
\[
\operatorname*{E}_{X\in\left\{  0,1\right\}  ^{N}}\left[  (\widetilde{p}%
(X)-p(X))^{2}\right]  \leq\alpha^{2}\delta^{2}.
\]
By the Cauchy-Schwarz inequality, then,%
\[
\operatorname*{E}_{X\in\left\{  0,1\right\}  ^{N}}\left[  \left\vert
\widetilde{p}(X)-p(X)\right\vert \right]  \leq\alpha\delta,
\]
so by Markov's inequality%
\[
\Pr_{X\in\left\{  0,1\right\}  ^{N}}\left[  \left\vert \widetilde{p}%
(X)-p(X)\right\vert >\alpha\right]  <\delta.
\]
Thus, our algorithm is simply to query $x_{i_{1}},\ldots,x_{i_{K}}$, and then
output $\widetilde{p}(X)$\ as our estimate for $p(X)$.
\end{proof}

Likewise:

\begin{corollary}
\label{dqcor2}$\operatorname*{D}_{\varepsilon+\delta}(f)\leq
2^{O(\operatorname*{Q}_{\varepsilon}(f))}/\delta^{4}$ for all Boolean
functions $f$\ and all $\varepsilon,\delta>0$.
\end{corollary}

\begin{proof}
Set $\alpha$\ to any constant less than $1/6$, then use the algorithm of
Corollary \ref{nonadap}\ to simulate the $\varepsilon$-approximate\ quantum
algorithm for $f$. \ Output $f(X)=1$\ if\ $\widetilde{p}(X)\geq\frac{1}{2}%
$\ and $f(X)=0$\ otherwise.
\end{proof}

Given an oracle $A$, let $\mathsf{BQP}^{A\left[  \log\right]  }$\ be the class
of languages decidable by a $\mathsf{BQP}$\ machine able to make $O\left(
\log n\right)  $\ queries to $A$. \ Also, let $\mathsf{AvgP}_{||}^{A}$\ be the
class of languages decidable, with probability $1-o\left(  1\right)  $\ over
$x\in\left\{  0,1\right\}  ^{n}$, by a $\mathsf{P}$\ machine able to make
$\operatorname*{poly}(n)$\ parallel (nonadaptive) queries to $A$. \ Then we
get the following unconditional variant of Theorem \ref{pbqpthm}.

\begin{theorem}
\label{bqplogthm}Suppose $\mathsf{P}=\mathsf{P}^{\mathsf{\#P}}$. \ Then
$\mathsf{BQP}^{A\left[  \log\right]  }\subset\mathsf{AvgP}_{||}^{A}$\ with
probability $1$ for a random oracle $A$.
\end{theorem}

\begin{proof}
The proof is essentially the same as that of Theorem \ref{pbqpthm}, except
that we use Corollary \ref{nonadap} in place of Conjecture \ref{infconj}. \ In
the proof of Corollary \ref{nonadap}, observe that the condition%
\[
\operatorname*{E}_{X\in\left\{  0,1\right\}  ^{N}}\left[  \left\vert
\widetilde{p}(X)-p(X)\right\vert \right]  \leq\alpha\delta
\]
implies%
\begin{equation}
\operatorname*{E}_{X\in\left\{  0,1\right\}  ^{N}}\left[  \left\vert p_{\mu
}(X)-p(X)\right\vert \right]  \leq\alpha\delta\label{eps2}%
\end{equation}
as well, where $p_{\mu}(X)$\ equals the mean of $p(Y)$\ over all inputs
$Y$\ that agree with $X$\ on $x_{i_{1}},\ldots,x_{i_{K}}$. \ Thus, given a
quantum algorithm that makes $T$ queries to an oracle string, the
computational problem that we need to solve boils down to finding a subset of
the oracle bits $x_{i_{1}},\ldots,x_{i_{K}}$\ such that\ $K=\frac{2^{O\left(
T\right)  }}{\alpha^{4}\delta^{4}}$\ and (\ref{eps2}) holds. \ Just like in
Theorem \ref{pbqpthm}, this problem is solvable in the counting hierarchy
$\mathsf{CH}=\mathsf{P}^{\mathsf{\#P}}\cup\mathsf{P}^{\mathsf{\#P}%
^{\mathsf{\#P}}}\cup\cdots$. \ So if we assume $\mathsf{P}=\mathsf{P}%
^{\mathsf{\#P}}$, then it is also solvable in $\mathsf{P}$.

In Theorem \ref{pbqpthm}, the conclusion we got was $\mathsf{BQP}^{A}%
\subset\mathsf{AvgP}^{A}$\ \ with probability $1$ for a random oracle $A$.
\ In our case, the number of classical queries $K$\ is exponential (rather
than polynomial) in the number of quantum queries $T$, so we only get
$\mathsf{BQP}^{A\left[  \log\right]  }\subset\mathsf{AvgP}^{A}$. \ On the
other hand, since the classical queries are nonadaptive, we can strengthen the
conclusion to $\mathsf{BQP}^{A\left[  \log\right]  }\subset\mathsf{AvgP}%
_{||}^{A}$.
\end{proof}

\section{Open Problems\label{OPEN}}

It would be nice to improve the $\operatorname*{R}(f)=O(\operatorname*{Q}%
(f)^{7}\operatorname*{polylog}\operatorname*{Q}(f))$\ bound for all symmetric
problems. \ As mentioned earlier, we conjecture that the right answer is
$\operatorname*{R}(f)=O(\operatorname*{Q}(f)^{2})$. \ In trying to improve our
lower bound, it seems best to avoid the use of \textsc{SetEquality}. \ After
all, it is a curious feature of our proof that, to get a lower bound for
symmetric problems, we need to reduce from the \textit{non}-symmetric
\textsc{SetEquality}\ problem!

Another problem is to remove the assumption $M\geq N$\ in our lower bound for
symmetric problems. \ Experience with related problems strongly suggests that
this can be done, but one might need to replace our chopping procedure by
something different.

We also conjecture that $\operatorname*{R}(f)\leq\operatorname*{Q}(f)^{O(1)}%
$\ for all partial functions $f$ that are symmetric \textit{only} under
permuting the inputs (and not necessarily the outputs). \ Proving this seems
to require a new approach. \ Another problem, in a similar spirit, is whether
$\operatorname*{R}(f)\leq\operatorname*{Q}(f)^{O(1)}$\ for all partial
functions $f:S\rightarrow\left\{  0,1\right\}  $ such that $S$ (i.e., the
promise on inputs) is symmetric, but $f$ itself need not be symmetric.

It would be interesting to reprove the $\operatorname*{R}(f)\leq
\operatorname*{Q}(f)^{O(1)}$\ bound using only the polynomial method, and not
the adversary method. \ Or, to rephrase this as a purely classical question:
for all $X=(x_{1},\ldots,x_{N})$ in $\left[  M\right]  ^{N}$, let $B_{X}$\ be
the $N\times M$ matrix whose $\left(  i,j\right)  ^{th}$\ entry is $1$\ if
$x_{i}=j$ and $0$ otherwise. \ Then given a set $S\subseteq\left[  M\right]
^{N}$\ and a function $f:S\rightarrow\left\{  0,1\right\}  $, let
$\widetilde{\deg}(f)$\ be the minimum degree of a real polynomial
$p:\mathbb{R}^{MN}\rightarrow\mathbb{R}$\ such that

\begin{enumerate}
\item[(i)] $0\leq p(B_{X})\leq1$ for all $X\in\left[  M\right]  ^{N}$, and

\item[(ii)] $\left\vert p(B_{X})-f(X)\right\vert \leq\frac{1}{3}$ for all
$X\in S$.
\end{enumerate}

\noindent Then is it the case that $\operatorname*{R}(f)\leq\widetilde{\deg
}(f)^{O\left(  1\right)  }$\ for all permutation-invariant functions $f$?

On the random oracle side, the obvious problem is to prove Conjecture
\ref{infconj}---thereby establishing that $\operatorname*{D}_{\varepsilon}(f)$
and $\operatorname*{Q}_{\delta}(f)$ are polynomially related, and all the
other consequences shown in Section \ref{RO}. \ Alternatively, one could look
for some technique that was tailored to polynomials $p$ that arise as the
acceptance probabilities of quantum algorithms. \ In this way, one could
conceivably solve $\operatorname*{D}_{\varepsilon}(f)$ versus
$\operatorname*{Q}_{\delta}(f)$\ and the other quantum problems, without
settling the general conjecture about bounded polynomials.

\section{Acknowledgments}

We thank Aleksandrs Belovs, Andy Drucker, Ryan O'Donnell, and Ronald de Wolf
for helpful discussions; Mark Zhandry for taking up our challenge to improve
the lower bound on $\operatorname*{Q}\left(  \text{\textsc{SetEquality}%
}\right)  $\ to the optimal $\Omega(N^{1/3})$; and Dana Moshkovitz for
suggesting a proof of Lemma \ref{bepsilon}. \ We especially thank Art\={u}rs
Ba\v{c}kurs, J\={a}nis Iraids, the attendees of the quantum computing reading
group at the University of Latvia, and the anonymous reviewers for their
feedback, and for catching some errors in earlier versions of this paper.

\bibliographystyle{plain}
\bibliography{thesis}

\section{Appendix: The Boolean Case\label{BOOLEAN}}

Given a partial Boolean function $f:\left\{  0,1\right\}  ^{N}\rightarrow
\left\{  0,1,\ast\right\}  $, call $f$ \textit{symmetric} if $f(X)$\ depends
only on the Hamming weight $\left\vert X\right\vert :=x_{1}+\cdots+x_{N}$.
\ For completeness, in this appendix we prove the following basic fact:

\begin{theorem}
\label{booleanthm}$\operatorname*{R}(f)=O(\operatorname*{Q}(f)^{2})$ for every
partial symmetric Boolean function $f$.
\end{theorem}

For \textit{total} symmetric Boolean functions, Theorem \ref{booleanthm}\ was
already shown by Beals et al.\ \cite{bbcmw}, using an approximation theory
result of Paturi \cite{paturi}. \ Indeed, in the total case one even has
$\operatorname*{D}(f)=O(\operatorname*{Q}(f)^{2})$. \ So the new twist is just
that $f$ can be partial.

Abusing notation, let $f\left(  k\right)  \in\left\{  0,1,\ast\right\}  $\ be
the value of $f$ on all inputs of Hamming weight $k$ (where as usual, $\ast
$\ means `undefined'). \ Then we have the following quantum lower bound:

\begin{lemma}
\label{booleanqlb}Suppose that $f\left(  a\right)  =0$\ and $f\left(
b\right)  =1$\ or vice versa, where $a<b$\ and $a\leq N/2$. \ Then
$\operatorname*{Q}(f)=\Omega\left(  \frac{\sqrt{bN}}{b-a}\right)  $.
\end{lemma}

\begin{proof}
This follows from a straightforward application of Ambainis's adversary
theorem (Theorem \ref{ambthm}). \ Specifically, let $A,B\subseteq\left\{
0,1\right\}  ^{N}$\ be the sets of all strings of Hamming weights $a$\ and $b$
respectively, and for all $X\in A$\ and $Y\in B$, put $\left(  X,Y\right)  \in
R$\ if and only if $X\preceq Y$ (that is, $x_{i}\leq y_{i}$\ for all
$i\in\left[  N\right]  $). \ Then%
\[
\operatorname*{Q}\left(  f\right)  =\Omega\left(  \sqrt{\frac{N-a}{b-a}%
\cdot\frac{b}{b-a}}\right)  =\Omega\left(  \frac{\sqrt{bN}}{b-a}\right)  .
\]
Alternatively, this lemma can be proved using the approximation theory result
of Paturi \cite{paturi}, following Beals et al.\ \cite{bbcmw}.
\end{proof}

In particular, if we set $\beta:=\frac{b}{N}$\ and $\varepsilon:=\frac{b-a}%
{N}$,\ then $\operatorname*{Q}(f)=\Omega(\sqrt{\beta}/\varepsilon)$. \ On the
other hand, we also have the following randomized \textit{upper} bound, which
follows from a Chernoff bound (similar to Lemma \ref{sampling}):

\begin{lemma}
\label{booleanrub}Assume $\beta>\varepsilon>0$. \ By making $O(\beta
/\varepsilon^{2})$\ queries to an $N$-bit string $X$, a classical sampling
algorithm can estimate the fraction $\beta:=\left\vert X\right\vert /N$\ of
$1$ bits to within an additive error $\pm\varepsilon/3$, with success
probability at least $2/3$.
\end{lemma}

Thus,\ assume the function $f$ is non-constant, and let%
\begin{equation}
\gamma:=\max_{f\left(  a\right)  =0,f\left(  b\right)  =1}\frac{\sqrt{bN}%
}{b-a}. \label{gamma}%
\end{equation}
Assume without loss of generality that the maximum of (\ref{gamma}) is
achieved when $a<b$\ and $a\leq N/2$, if necessary by applying the
transformations $f(X)\rightarrow1-f(X)$\ and $f(X)\rightarrow f(N-X)$.\ \ Now
consider the following randomized algorithm to evaluate $f$, which makes
$T:=O(\gamma^{2})$\ queries:\bigskip

\texttt{\qquad Choose indices }$i_{1},\ldots,i_{T}\in\left[  N\right]
$\texttt{\ uniformly at random with replacement}

\texttt{\qquad Query }$x_{i_{1}},\ldots,x_{i_{T}}$

\texttt{\qquad Set }$k:=\frac{N}{T}(x_{i_{1}}+\cdots+x_{i_{T}})$

\texttt{\qquad If there exists a }$b\in\left\{  0,\ldots,N\right\}
$\texttt{\ such that }$f\left(  b\right)  =1$\texttt{\ and }$\left\vert
k-b\right\vert \leq\frac{\sqrt{bN}}{3\gamma}$

\texttt{\qquad\qquad output }$f(X)=1$

\texttt{\qquad Otherwise output }$f(X)=0\bigskip$

By Lemma \ref{booleanrub}, the above algorithm succeeds with probability at
least $2/3$, provided we choose $T$\ suitably large. \ Hence
$\operatorname*{R}(f)=O(\gamma^{2})$. \ On the other hand, Lemma
\ref{booleanqlb} implies that $\operatorname*{Q}(f)=\Omega(\gamma)$. \ Hence
$\operatorname*{R}(f)=O(\operatorname*{Q}(f)^{2})$, completing the proof of
Theorem \ref{booleanthm}.

\section{Appendix: $1$-Norm versus $2$-Norm\label{NORMS}}

As mentioned in Section \ref{SUBSEQ}, in the original version of this paper we
stated Conjecture \ref{infconj}, and all our results assuming it, in terms of
$L_{1}$-influences rather than $L_{2}$-influences. \ Subsequently, Arturs
Ba\v{c}kurs discovered a gap in our $L_{1}$-based\ argument. \ In recent work,
Ba\v{c}kurs and Bavarian \cite{bacbav} managed to fill the gap, allowing our
$L_{1}$-based\ argument to proceed. \ Still, the \textit{simplest} fix for the
problem Ba\v{c}kurs\ uncovered is just to switch from $L_{1}$-influences to
$L_{2}$-influences, so that is what we did in Section \ref{RO}\ (and in our
current statement of Conjecture \ref{infconj}).

Fortunately, it turns out that the $L_{1}$\ and $L_{2}$\ versions of
Conjecture \ref{infconj} are \textit{equivalent}, so making this change does
not even involve changing our conjecture. \ For completeness, in this appendix
we prove the equivalence of the $L_{1}$\ and $L_{2}$\ versions of Conjecture
\ref{infconj}.

As usual, let $p:\left\{  0,1\right\}  ^{N}\rightarrow\left[  0,1\right]
$\ be a real polynomial,\ let $X\in\left\{  0,1\right\}  ^{N}$, and let
$X^{i}$\ denote $X$\ with the $i^{th}$\ bit flipped. \ Then the $L_{1}%
$-\textit{variance} $\operatorname*{Vr}\left[  p\right]  $ of $p$ and the
$L_{1}$-\textit{influence} $\operatorname*{Inf}\nolimits_{i}\left[  p\right]
$\ of the $i^{th}$\ variable $x_{i}$ are defined as follows:%
\begin{align*}
\operatorname*{Vr}\left[  p\right]   &  :=\operatorname*{E}_{X\in\left\{
0,1\right\}  ^{N}}\left[  \left\vert p(X)-\operatorname*{E}\left[  p\right]
\right\vert \right]  ,\\
\operatorname*{Inf}\nolimits_{i}^{1}\left[  p\right]   &  :=\operatorname*{E}%
_{X\in\left\{  0,1\right\}  ^{N}}\left[  \left\vert p(X)-p(X^{i})\right\vert
\right]  .
\end{align*}
The $L_{1}$\ analogue of Conjecture \ref{infconj} simply replaces
$\operatorname*{Var}\left[  p\right]  $\ by $\operatorname*{Vr}\left[
p\right]  $\ and $\operatorname*{Inf}\nolimits_{i}\left[  p\right]  $\ by
$\operatorname*{Inf}\nolimits_{i}^{1}\left[  p\right]  $:

\begin{conjecture}
[Bounded Polynomials Have Influential Variables, $L_{1}$ Version]%
\label{infconj1}Let $p:\mathbb{R}^{N}\rightarrow\mathbb{R}$\ be a degree-$d$
real polynomial such that $0\leq p(X)\leq1$\ for all $X\in\left\{
0,1\right\}  ^{N}$. \ Then there exists an $i\in\left[  N\right]  $ such that
$\operatorname*{Inf}\nolimits_{i}^{1}\left[  p\right]  \geq(\operatorname*{Vr}%
\left[  p\right]  /d)^{O\left(  1\right)  }$.
\end{conjecture}

We now prove the equivalence:

\begin{proposition}
\label{equivprop}Conjectures \ref{infconj} and \ref{infconj1}\ are equivalent.
\end{proposition}

\begin{proof}
First assume Conjecture \ref{infconj}. \ By the Cauchy-Schwarz inequality,%
\[
\operatorname*{Inf}\nolimits_{i}\left[  p\right]  =\operatorname*{E}%
_{X\in\left\{  0,1\right\}  ^{N}}\left[  (p(X)-p(X^{i}))^{2}\right]
\geq\left(  \operatorname*{E}_{X\in\left\{  0,1\right\}  ^{N}}\left[
\left\vert p(X)-p(X^{i})\right\vert \right]  \right)  ^{2}=\operatorname*{Inf}%
\nolimits_{i}^{1}\left[  p\right]  ^{2}.
\]
Also, since $p(X)\in\left[  0,1\right]  $,%
\[
\operatorname*{Vr}\left[  p\right]  =\operatorname*{E}_{X\in\left\{
0,1\right\}  ^{N}}\left[  \left\vert p(X)-\operatorname*{E}\left[  p\right]
\right\vert \right]  \geq\operatorname*{E}_{X\in\left\{  0,1\right\}  ^{N}%
}\left[  (p(X)-\operatorname*{E}\left[  p\right]  )^{2}\right]
=\operatorname*{Var}\left[  p\right]  .
\]
Hence there exists an $i\in\left[  N\right]  $\ such that%
\[
\operatorname*{Inf}\nolimits_{i}\left[  p\right]  \geq\operatorname*{Inf}%
\nolimits_{i}^{1}\left[  p\right]  ^{2}\geq\left(  \frac{\operatorname*{Vr}%
\left[  p\right]  }{d}\right)  ^{O\left(  1\right)  }\geq\left(
\frac{\operatorname*{Var}\left[  p\right]  }{d}\right)  ^{O\left(  1\right)  }%
\]
and Conjecture \ref{infconj1}\ holds.

Likewise, assume Conjecture \ref{infconj1}. \ Then we have
$\operatorname*{Inf}\nolimits_{i}^{1}\left[  p\right]  \geq\operatorname*{Inf}%
\nolimits_{i}\left[  p\right]  $\ since $p(X)\in\left[  0,1\right]  $, and
$\operatorname*{Var}\left[  p\right]  \geq\operatorname*{Vr}\left[  p\right]
^{2}$ by the Cauchy-Schwarz inequality. \ Hence there exists an $i\in\left[
N\right]  $\ such that%
\[
\operatorname*{Inf}\nolimits_{i}^{1}\left[  p\right]  \geq\operatorname*{Inf}%
\nolimits_{i}\left[  p\right]  \geq\left(  \frac{\operatorname*{Var}\left[
p\right]  }{d}\right)  ^{O\left(  1\right)  }\geq\left(  \frac
{\operatorname*{Vr}\left[  p\right]  }{d}\right)  ^{O\left(  1\right)  }%
\]
and Conjecture \ref{infconj}\ holds.
\end{proof}

\section{\label{EQUIV}Appendix: Equivalent Form of Conjecture
\ref{folkloreconj}}

Recall Conjecture \ref{folkloreconj}, which said (informally) that any quantum
algorithm that makes $T$ queries to $X\in\left\{  0,1\right\}  ^{N}$ can be
simulated to within $\pm\varepsilon$\ additive error on a $1-\delta$\ fraction
of $X$'s by a classical algorithm that makes $\operatorname*{poly}%
(T,1/\varepsilon,1/\delta)$\ queries. \ In Section \ref{FORMAL}, we claimed
that Conjecture \ref{folkloreconj}\ was equivalent to an alternative
conjecture, which we now state more formally:

\begin{conjecture}
\label{folkloreconj2}Let $S\subseteq\left\{  0,1\right\}  ^{N}$\ with
$\left\vert S\right\vert \geq c2^{N}$, and let $f:S\rightarrow\left\{
0,1\right\}  $. \ Then there exists a deterministic classical algorithm that
makes $\operatorname*{poly}(\operatorname*{Q}(f),1/\alpha,1/c)$\ queries, and
that computes $f(X)$\ on at least a $1-\alpha$\ fraction of $X\in S$.
\end{conjecture}

In this appendix, we justify the equivalence claim. \ We first need a simple
combinatorial lemma.

\begin{lemma}
\label{bepsilon}Suppose we are trying to learn an unknown real $p\in\left[
0,1\right]  $. \ There are $k$ \textquotedblleft hint bits\textquotedblright%
\ $h_{1},\ldots,h_{k}$, where each $h_{i}$\ is $0$ if $\left(  i-1\right)
/k\leq p$ or $1$\ if $i/k\geq p$ (and can otherwise be arbitrary). \ However,
at most $b<k/2$ of the $h_{i}$'s are then corrupted by an adversary, producing
the new string $h_{1}^{\prime},\ldots,h_{k}^{\prime}$. \ Using $h_{1}^{\prime
},\ldots,h_{k}^{\prime}$, one can still determine $p$\ to within additive
error $\pm\left(  b+1\right)  \varepsilon$.
\end{lemma}

\begin{proof}
Given the string $h^{\prime}=\left(  h_{1}^{\prime},\ldots,h_{k}^{\prime
}\right)  $, we apply the following correction procedure: we repeatedly search
for pairs $i<j$\ such that $h_{i}^{\prime}=1$\ and $h_{j}^{\prime}=0$, and
\textquotedblleft delete\textquotedblright\ those pairs (that is, we set
$h_{i}^{\prime}=h_{j}^{\prime}=\ast$, where $\ast$\ means \textquotedblleft
unknown\textquotedblright). \ We continue for $t$ steps, until no more such
pairs exist. \ Next, we delete the rightmost $b-t$\ zeroes in $h^{\prime}%
$\ (replacing them with $\ast$'s), and likewise delete the leftmost
$b-t$\ ones.\ \ Finally, as our estimate for $p$, we output%
\[
q:=\frac{i^{\ast}+j^{\ast}-1}{2k},
\]
where $i^{\ast}$\ is the index of the rightmost $0$ remaining in $h^{\prime}%
$\ (or $i^{\ast}=0$ if no $0$'s remain), and $j^{\ast}$ is the index of the
leftmost $1$ remaining (or $j^{\ast}=k+1$\ if no $1$'s\ remain).

To show correctness: every time we find an $i<j$ pair such that $h_{i}%
^{\prime}=1$\ and $h_{j}^{\prime}=0$, at least one of $h_{i}^{\prime}$\ and
$h_{j}^{\prime}$\ must have been corrupted by the adversary. \ It follows that
$t\leq b$, where $t$ is the number of deleted pairs. \ Furthermore, after the
first stage finishes, every $1$ is to the right of every $0$, at most
$b-t$\ of the remaining bits are corrupted, and the bits that \textit{are}
corrupted must be among the rightmost zeroes of the leftmost ones (or both).
\ Hence, after the second stage finishes, every $h_{i}^{\prime}=0$\ reliably
indicates that $p\geq\left(  i-1\right)  /k$, and every $h_{j}^{\prime}%
=1$\ reliably indicates that $p\leq j/k$. \ Moreover, since only $2b$\ bits
are deleted in total, we must have $j^{\ast}-i^{\ast}\leq2b+1$, where
$i^{\ast}$ and $j^{\ast}$ are as defined above. \ It follows that $\left\vert
p-q\right\vert \leq\left(  b+1\right)  \varepsilon$.
\end{proof}

\begin{theorem}
\label{equivthm2}Conjectures \ref{folkloreconj}\ and \ref{folkloreconj2}\ are equivalent.
\end{theorem}

\begin{proof}
We start with the easy direction, that Conjecture \ref{folkloreconj}\ implies
Conjecture \ref{folkloreconj2}. \ Given $f:S\rightarrow\left\{  0,1\right\}  $
with $\left\vert S\right\vert \geq c2^{N}$, let $Q$\ be a quantum algorithm
that evaluates $f$ with error probability at most $1/3$\ using $T$ queries.
\ Let $p(X)$\ be $Q$'s acceptance probability on a given input $X\in\left\{
0,1\right\}  ^{N}$\ (not necessarily in $S$). \ Then by Conjecture
\ref{folkloreconj}, there exists a deterministic classical algorithm that
approximates $p(X)$\ to within additive error $\pm\varepsilon$\ on a
$1-\delta$\ fraction of $X\in\left\{  0,1\right\}  ^{N}$\ using
$\operatorname*{poly}(T,1/\varepsilon,1/\delta)$\ queries. \ If we set (say)
$\varepsilon:=1/7$\ and $\delta:=\alpha c$, then such an approximation lets us
decide whether $f(X)=0$\ or $f(X)=1$\ for a $1-\alpha$\ fraction of $X\in S$,
using $\operatorname*{poly}(T,1/\alpha,1/c)$\ queries.

We now show the other direction, that Conjecture \ref{folkloreconj2} implies
Conjecture \ref{folkloreconj}. \ Let $Q$ be a $T$-query quantum algorithm, let
$p(X)$\ be $Q$'s acceptance probability on input $X$, and suppose we want to
approximate $p(X)$\ to within error\ $\pm\varepsilon$\ on at least a
$1-\delta$\ fraction of $X\in\left\{  0,1\right\}  ^{N}$. \ Let $\epsilon
:=\varepsilon/3$. \ Assume for simplicity that $\epsilon$\ has the form
$1/k$\ for some positive integer $k$; this will have no effect on the
asymptotics. \ For each $j\in\left[  k\right]  $, let%
\[
S_{j}:=\left\{  X:p(X)\leq\frac{j-1}{k}\text{ or }p(X)\geq\frac{j}{k}\right\}
,
\]
and define the function $f_{j}:S_{j}\rightarrow\left\{  0,1\right\}  $\ by%
\[
f_{j}(X):=\left\{
\begin{array}
[c]{cc}%
0 & \text{if }p(X)\leq\left(  j-1\right)  /k\\
1 & \text{if }p(X)\geq j/k.
\end{array}
\right.
\]
By Proposition \ref{amplify}, we have $\operatorname*{Q}(f_{j})=O(kT)$ for all
$j\in\left[  k\right]  $. \ Also, note that%
\[
\operatorname*{E}_{j}\left[  \left\vert S_{j}\right\vert \right]  \geq\left(
1-\frac{1}{k}\right)  2^{n}.
\]
By Markov's inequality, this implies that there can be at most one
$j\in\left[  k\right]  $\ (call it $j^{\ast}$) such that $\left\vert
S_{j}\right\vert <2^{n-2}$. \ Likewise, note that for every $X\in\left\{
0,1\right\}  ^{N}$, there is at most one $j\in\left[  k\right]  $\ such that
$X\notin S_{j}$.

Together with Conjecture \ref{folkloreconj2}, the above facts imply that, for
all $j\neq j^{\ast}$ and $\alpha>0$, there exists a deterministic classical
algorithm $A_{j,\alpha}$, making $\operatorname*{poly}(T,1/\alpha)$\ queries,
that computes $f_{j}(X)$\ on at least a $1-\alpha$\ fraction of all $X\in
S_{j}$. \ Suppose we run $A_{j,\alpha}$\ for all $j\neq j^{\ast}$. \ Then by
the union bound, for at least a $1-k\alpha$\ fraction of $X\in\left\{
0,1\right\}  ^{N}$, there can be at most two $j\in\left[  k\right]  $\ such
that $A_{j,\alpha}$\ fails to compute $f_{j}(X)$: namely, $j^{\ast}$, and the
unique $j$ (call it $j^{\prime}$) such that $X\notin S_{j^{\prime}}$. \ Thus,
suppose $A_{j,\alpha}$\ succeeds for all $j\notin\left\{  j^{\ast},j^{\prime
}\right\}  $. \ By Lemma \ref{bepsilon}, this implies that $p(X)$\ has been
determined up to an additive error of $\pm3\epsilon=\pm\varepsilon$. \ Hence,
we simply need to set $\alpha:=\delta/k$, in order to get a classical
algorithm that makes $k\cdot\operatorname*{poly}(T,k/\delta
)=\operatorname*{poly}(T,1/\varepsilon,1/\delta)$ queries, and that
approximates $p(X)$\ up to additive error $\pm\varepsilon$\ for at least
a\ $1-\delta$\ fraction of $X\in\left\{  0,1\right\}  ^{N}$.
\end{proof}

\end{document}